%% file: main.tex
\documentclass[11pt,onecolumn]{IEEEtran}

\usepackage{amsmath, amssymb, amsfonts, amsthm, amstext}
\usepackage{bm}
\usepackage{bbm}
\usepackage{mathrsfs}

\usepackage{graphicx}
\usepackage{epstopdf}
\usepackage[table,xcdraw]{xcolor}
\usepackage{array}
\usepackage{tabularx}
\usepackage{multirow}
\usepackage{multicol}
\usepackage[caption=false, font=footnotesize]{subfig}
\usepackage{adjustbox}
\usepackage{booktabs}
\usepackage{etoolbox}
\makeatletter
\patchcmd{\@makecaption}
  {\@IEEEtabletopskipstrut{\normalfont\footnotesize #1}\\{\normalfont\footnotesize\scshape #2}}
  {\@IEEEtabletopskipstrut{\normalfont\footnotesize #1\nobreakspace{\normalfont\footnotesize\scshape #2}}}
  {}{}
\makeatother

\usepackage[linesnumbered,ruled,vlined]{algorithm2e}
\SetKwComment{Comment}{\color{blue} $\triangleright$\ }{}
\SetKwInOut{KwIn}{Input}
\SetKwInOut{KwOut}{Output}
\SetKwProg{Proc}{Subroutine}{}{}

\usepackage{indentfirst}
\usepackage{verbatim}
\usepackage{stfloats}
\usepackage{float}
\usepackage{placeins}
\usepackage{color}
\usepackage{cite}
\usepackage[colorlinks=true,linkcolor=blue,citecolor=blue]{hyperref}
\usepackage{enumitem}
\usepackage{scalerel}
\usepackage[norefs,nocites]{refcheck}

\newcommand{\cA}{\mathcal{A}}

\newcommand{\cD}{\mathcal{D}}

\newcommand{\cG}{\mathcal{G}}

\newcommand{\cL}{\mathcal{L}}
\newcommand{\cM}{\mathcal{M}}

\newcommand{\cT}{\mathcal{T}}

\newcommand{\cZ}{\mathcal{Z}}

\theoremstyle{definition}
\newtheorem{theorem}{\bf Theorem}
\newtheorem{lemma}{\bf Lemma}
\newtheorem{proposition}{\bf Proposition}
\newtheorem{corollary}{\bf Corollary}

\newtheorem{definition}{\bf Definition}

\newtheorem{example}{\bf Example}
\newtheorem{construction}{\bf Construction}

\newcommand{\Sn}{\mathcal{S}}
\newcommand{\std}{\operatorname{std}}
\newcommand{\red}{\operatorname{red}}
\newcommand{\Enc}{\operatorname{Enc}}
\newcommand{\LCS}{\operatorname{LCS}}
\newcommand{\LPat}{\operatorname{LPat}}
\newcommand{\SID}{\mathrm{SID}}
\newcommand{\PID}{\mathrm{PID}}
\newcommand{\Rep}{\operatorname{Rep}}
\newcommand{\rk}{\operatorname{rk}}
\newcommand{\VT}{\operatorname{VT}}
\newcommand{\Int}{\operatorname{Int}}

\newcommand{\mstd}{\operatorname{mstd}}
\newcommand{\Lift}{\operatorname{Lift}}
\newcommand{\supp}{\operatorname{supp}}

\DeclareRobustCommand\revised{\textcolor{red}}

\begin{document}

\title{Order-Optimal Systematic Permutation Codes for Correcting $t$ Deletions}

\author{Bolin Wu, Quan D. Bui,
Kai~Niu, Van Khu Vu, Shuche Wang
\thanks
	{
		Bolin Wu is with the School of Mathematical Sciences, Beijing University of Posts and Telecommunications, Beijing 100876, China. (e-mail: bolinwu@bupt.edu.cn.)  \\ \indent
    Van Khu Vu and Quan D. Bui are with VinUniversity, Gia Lam, Hanoi 12400, Vietnam. (e-mail: \{khu.vv, quan.bd\}@vinuni.edu.vn.) \\ \indent
		Kai Niu is with the Key Laboratory of Universal Wireless Communications, Ministry of Education, Beijing University of Posts and Telecommunications, Beijing 100876, China. (e-mail: niukai@bupt.edu.cn.)\\ \indent
		Shuche Wang is with the Department of Mathematics, National University of Singapore, Singapore 119076 (e-mail: shuche.wang@u.nus.edu).\\ \indent
    (\emph{Corresponding Author: Shuche Wang}.)
}
}

\maketitle

\begin{abstract}
This paper investigates the construction of full-systematic permutation codes capable of correcting multiple deletions under two complementary models, namely symbol-invariant deletions (SIDs), where surviving symbol values are preserved, and permutation-invariant deletions (PIDs), where the surviving sequence is standardized to a permutation. For any fixed integer $t \ge 1$ and all sufficiently large message lengths $n$, our proposed encoders map any message permutation of length $n$ to a codeword by inserting distinct redundancy symbols while strictly preserving the sequence order of the original message symbols. The proposed constructions correct up to $t$ deletions using $7t-1$ redundancy markers for PIDs and $4t$ redundancy markers for SIDs, achieving redundancies of $(7t-1)\log n + O_t(1)$ bits and $4t\log n + O_t(1)$ bits, respectively. Both code families are uniformly constructible, encodable, and decodable in $n^{O(t)}$ time. The underlying framework stores an inner deletion-correcting syndrome in the relative positions of redundancy markers via an algebraic outer code based on integer moments and residual graph coloring. We further extend this framework to fixed-composition and strictly $\lambda$-regular multipermutations, proving that the PID and SID channels coincide whenever the common multiplicity satisfies $\lambda > t$. 
\end{abstract}

\input{main_paper/introduction}

\input{main_paper/notation}
\input{main_paper/pid}
\input{main_paper/sid}
\input{main_paper/multipermutation}
\input{main_paper/comparison}
\input{main_paper/conclusion}

\newpage
\appendices

\input{Appendix/lower_bound}

\bibliographystyle{IEEEtran}
\bibliography{ref}

\end{document}

%% file: main_paper/introduction.tex
\section{Introduction}
\label{sec:introduction}

In non-volatile flash memory systems, the rank modulation scheme introduced by Jiang \emph{et al.}~\cite{jiang2009rank} represents information through the relative order of cell charge levels rather than their absolute physical voltages. By encoding data as permutations, rank modulation inherently eliminates programming overshoots and provides strong resilience against retention noise caused by gradual charge leakage~\cite{jiang2010correcting,barg2010codes}. When stored cells experience severe hardware defects, localized cell damage, or sensing failures, the stored permutations suffer from cell erasures or deletions.

In their foundational work on rank modulation, Gabrys \emph{et al.}~\cite{gabrys2015codes} formalized two distinct deletion error models, termed symbol-invariant deletions (SIDs) and permutation-invariant deletions (PIDs). Under an SID, also known as a stable deletion, lost symbols are removed while the absolute labels of the surviving symbols remain unchanged. The problem of correcting SIDs is intrinsically linked to designing permutation codes in the Ulam metric~\cite{farnoud2013error,gabrys2015codes}. Under a PID, the reading mechanism discerns only the relative ordering of the surviving cells, so that the output sequence is standardized by replacing each entry with its rank among the remaining symbols. For example, deleting the symbol $4$ from the permutation $(3,1,4,2)$ yields $(3,1,2)$ under both error models. In contrast, deleting the symbol $1$ produces the sequence $(3,4,2)$ under an SID, but standardizes to $(2,3,1)$ under a PID. Because a PID conceals both the deletion coordinates and the underlying symbol values, designing algebraic codes for the PID channel requires distinct combinatorial mechanisms.

This paper investigates the algebraic construction of full-systematic permutation codes capable of correcting multiple deletions under both the SID and PID models. Let $\mathcal S_n$ denote the symmetric group on the set $[n] \triangleq \{1,\ldots,n\}$. In classical deletion channels over fixed alphabets, systematic codes are traditionally obtained by appending parity-check suffixes to the message bits~\cite{levenshtein1966binary,tenengolts1984nonbinary,sima2020optimalbinary,guruswami2021two}. In the permutation setting, however, appending a conventional checksum suffix is impossible because each symbol must appear exactly once. Appending symbols directly destroys the permutation structure and alters the rank of message symbols under PID standardization. 

To overcome this fundamental challenge, we define a full-systematic permutation code of length $n+r$ as one where every message permutation $\mu\in\mathcal S_n$ appears as an unaltered subsequence, with the remaining $r$ coordinates occupied by distinct redundancy markers chosen from $[n+1,n+r]$. Deleting these $r$ redundancy symbols from an uncorrupted codeword immediately retrieves the original message $\mu$. Such a code contains exactly $n!$ codewords in $\mathcal S_{n+r}$ and incurs a bit redundancy of $\log((n+r)!/n!) = r\log n + O_r(1)$. Our primary goal is to design full-systematic codes that correct up to a fixed number $t$ of arbitrary deletions using $r=O(t)$ redundancy markers, accompanied by polynomial-time encoding and decoding algorithms.

\subsection{Related Work and Problem Formulation}

Algebraic code constructions combating deletions have advanced significantly for sequence channels, where optimal systematic codes correcting multiple deletions were established by Sima, Gabrys, and Bruck~\cite{sima2020optimalbinary}, and explicit two-deletion codes were constructed by Guruswami and H{\aa}stad~\cite{guruswami2021two}. In the permutation domain, Levenshtein~\cite{levenshtein1992perfect} first connected stable deletions to codes in the Ulam metric, while Gabrys \emph{et al.}~\cite{gabrys2015codes} characterized the metric properties of SIDs and constructed single-PID-correcting codes. Subsequent works further addressed burst deletions in permutations~\cite{chee2019burst,sun2022improved}. More recently, Wang \emph{et al.}~\cite{wang2025permutation} improved the Gilbert-Varshamov bound in the Ulam metric and constructed multiple-deletion codes via a cyclic-successor mapping. However, all these existing permutation codes are non-systematic, meaning that message symbols do not appear directly as a subsequence and message retrieval requires full decoding. Moreover, codes capable of correcting multiple PIDs remain unexplored, as existing constructions are restricted either to a single deletion or to localized burst deletions.

Although systematic permutation codes have been investigated under the Kendall $\tau$ metric~\cite{zhou2014systematic} and Chebyshev metric~\cite{buzaglo2014perfect,klove2010permutation}, designing systematic permutation codes for deletion channels is substantially more difficult because conventional parity-check suffixes cannot simply be appended. Graph coloring offers a natural algebraic approach to this problem, where color classes of a confusability graph correspond directly to error-correcting codes. Recently, Li, Gabrys, and Farnoud~\cite{li2026distributed} introduced a distributed graph coloring framework that uses algebraic polynomial fingerprints over finite fields to efficiently compress proper colorings on graphs with bounded maximum degree. In this work, we adapt and specialize their algebraic recoloring technique to compress the residual confusability graphs arising from our single-PID decomposition, which enables the design of compact inner syndromes.

A broader generalization of practical importance is multipermutations, which model rank modulation schemes where multiple flash memory cells share identical charge levels~\cite{gad2012trade,huczynska2006frequency}. Farnoud and Milenkovic~\cite{hassanzadeh2014multipermutation} studied multipermutation codes under the Ulam metric, while Sala, Gabrys, and Dolecek~\cite{sala2014deletions} investigated deletion errors in multipermutations. However, existing multipermutation deletion codes are likewise non-systematic, and constructing systematic codes over arbitrary composition profiles remains completely unaddressed in the literature.

\subsection{Main Contributions}

For every fixed integer $t\ge1$ and all sufficiently large message lengths $n$, the main contributions of this paper are summarized in the following three points.

\begin{enumerate}[leftmargin=*,label=\arabic*)]
 \item We construct full-systematic codes on $\mathcal S_n$ that correct up to $t$ permutation-invariant deletions (PIDs) with redundancy $(7t-1)\log n+O_t(1)$ bits. Our construction employs a two-stage framework. In the inner stage, we decompose the $t$-PID confusability graph using the single-PID signature syndrome, which shrinks the maximum degree of the residual graph from $O_t(n^{3t})$ down to $O_t(n^{3t-2})$. Applying two rounds of local algebraic graph coloring compresses this coloring into an inner syndrome of size $O_t(n^{6t-2})$. In the outer stage, we represent the positions of $r=7t-1$ redundancy markers via a labeled-gap representation and construct an outer code for the projected PID channel using integer moment constraints combined with greedy residual coloring. This outer code achieves the optimal scaling exponent $r-t$, enabling unambiguous recovery of the inner syndrome. The resulting code family is uniformly constructible, encodable, and decodable in $n^{O(t)}$ time.

 \item We construct full-systematic codes on $\mathcal S_n$ that correct up to $t$ symbol-invariant deletions (SIDs) with redundancy $4t\log n+O_t(1)$ bits. In the inner stage, we apply a cyclic-successor mapping to embed the Ulam metric into the Hamming metric, yielding an algebraic power-sum syndrome over a finite field taking $O_t(n^{3t-1})$ values. For the outer stage, we establish an SID-to-PID reduction showing that the outer code developed for PIDs can be reused directly to protect the inner SID syndrome using $r=4t$ redundancy markers. Furthermore, by evaluating baseline codes obtained through unconstrained local graph fingerprinting, we prove that our structural inner constraints save $(t+1)\log n$ redundancy bits for SIDs and $2\log n$ redundancy bits for PIDs. In the special case $t=1$, the framework yields systematic codes using only three markers for SIDs and four markers for PIDs, achieving redundancies of $3\log n$ bits and $4\log n$ bits, respectively.

 \item We extend the systematic coding framework to multipermutations under both fixed-composition and strictly regular profiles. For repeated symbols, we define the multipermutation PID channel through tie-preserving standardization and construct a canonical stable lift that maps multipermutations to permutations while preserving deletion descendants. For any composition profile $\boldsymbol m$, we construct systematic codes using $5t$ markers for SIDs and $r_{\PID}(\boldsymbol m,t) = \min\{7t-1, 5t+\varepsilon(\boldsymbol m)+2\eta_t(\boldsymbol m)\}$ markers for PIDs, where $\eta_t(\boldsymbol m)$ and $\varepsilon(\boldsymbol m)$ capture the profile-dependent alphabet ambiguity. For strictly $\lambda$-regular multipermutations, we construct systematic encoders that preserve $\lambda$-regularity by adding $R = s\lambda$ occurrences. In particular, we prove that whenever $\lambda>t$, no symbol class can be completely eliminated by $t$ deletions, causing the PID and SID channels to coincide and eliminating the PID redundancy penalty. All proposed multipermutation code families admit polynomial-time construction, encoding, and decoding algorithms.
\end{enumerate}

\subsection{Organization of the Paper}

The remainder of this paper is organized as follows. Section~\ref{sec:preliminaries} provides necessary preliminaries, the deletion models, and code bounds. Sections~\ref{sec:pid} and~\ref{sec:sid} construct systematic codes for PIDs and SIDs, respectively. Section~\ref{sec:multipermutation} extends these constructions to multipermutations. Section~\ref{sec:comparison} evaluates redundancy and complexity trade-offs against baseline schemes. Finally, Section~\ref{sec:conclusion} concludes the paper.

%% file: main_paper/notation.tex
\section{Notation and Preliminaries}
\label{sec:preliminaries}

For a positive integer $m$, let $[m] \triangleq  \{1, \ldots, m\}$, and let $\Sn_m$ denote the set of permutations of $[m]$ in sequence notation. For integers $a \le b$, let $[a,b] \triangleq  \{a, a+1, \ldots, b\}$. We reserve $m$ for a generic codeword length, $n$ for the message length of a systematic encoder, and $r$ for the number of redundancy symbols. Throughout this paper, the error parameter $t \ge 1$ is fixed, all relevant blocklengths are strictly greater than $t$, and all logarithms are to the base two. The asymptotic notations $O_t(\cdot)$ and $\Omega_t(\cdot)$ suppress constants that may depend on $t$ but are independent of the blocklength.

Given a sequence $x = (x_1, \ldots, x_m)$ and an index set $I = \{i_1 < \cdots < i_s\} \subseteq [m]$, let $x_I \triangleq  (x_{i_1}, \ldots, x_{i_s})$. For an alphabet $A$, let $x|_A$ denote the subsequence of $x$ obtained by retaining the entries belonging to $A$, in their original order. We write
\begin{equation*}
  \mathcal I_{q,s} \triangleq  \{x \in [q]^s : x_i \ne x_j \text{ whenever } i \ne j\}
\end{equation*}
for the set of injective sequences of length $s$ over $[q]$; in particular, $\Sn_m = \mathcal I_{m,m}$. For $\pi \in \Sn_m$, the inverse permutation $\pi^{-1}$ is written in sequence notation, so that $\pi^{-1}_a$ denotes the position occupied by the symbol $a$ in $\pi$. For two sequences $x, y$ of the same length, their Hamming distance is
\begin{equation*}
  d_{\mathrm H}(x,y) \triangleq  |\{i : x_i \ne y_i\}|.
\end{equation*}
We write $\mathbb Z_q$ for the ring of integers modulo $q$, $\mathbb F_q$ for the finite field of order $q$ when $q$ is a prime power, and $\mathbf 1\{\mathcal E\}$ for the indicator function of an event $\mathcal E$. For a mapping $f$ from a set $A$ to a set $B$, let $\operatorname{im} (f) \triangleq  \{f(a) : a \in A\}$ denote the image of $f$.

\subsection{The SID and PID Deletion Models}\label{subsec:sidpid_model}

We adopt the symbol-invariant deletion (SID) and permutation-invariant deletion (PID) models introduced in~\cite{gabrys2015codes}. For a sequence $x = (x_1, \ldots, x_m)$ with pairwise distinct entries, its \emph{standardization} is the unique permutation $\std(x) \in \Sn_m$ defined by
\begin{equation*}
  \std(x)_i \triangleq 1 + |\{j \in [m] : x_j < x_i\}|, \quad i \in [m].
\end{equation*}
Hence, standardization preserves the coordinate order and the relative value order of the entries while bijectively mapping the alphabet to $[m]$.

\begin{definition}[Deletion Descendants]\label{def:deletion-models}
Let $x=(x_1,\ldots,x_m)$ and $I\subseteq[m]$. An SID, also called a stable deletion, retains the surviving labels, defined by
\begin{equation*}
\operatorname{del}^{\SID}_I(x) \triangleq x_{[m]\setminus I}.
\end{equation*}
If the entries of $x$ are pairwise distinct, a PID standardizes the surviving sequence according to
\begin{equation*}
  \operatorname{del}^{\PID}_I(x) \triangleq \std\!\left(x_{[m]\setminus I}\right).
\end{equation*}
For $\pi\in\Sn_m$, $\mathsf x\in\{\SID,\PID\}$, and $0\le s\le m$, let $\mathcal D_s^{\mathsf x}(\pi)$ denote the set of all descendants obtainable from $\pi$ after exactly $s$ $\mathsf x$-deletions, given by
\begin{equation*}
  \mathcal D_s^{\mathsf x}(\pi)
  \triangleq \left\{\operatorname{del}^{\mathsf x}_I(\pi):I\subseteq[m],\ |I|=s\right\}.
\end{equation*}
For $t\le m$, we further define $\mathcal D_{\le t}^{\mathsf x}(\pi) \triangleq \bigcup_{s=0}^t\mathcal D_s^{\mathsf x}(\pi)$.
\end{definition}

Unlike an SID, a PID standardizes the surviving sequence after the deletion operation. To compose multiple PID operations, it is necessary to verify that intermediate standardizations do not alter the relative order of entries retained subsequently. The following lemma establishes this consistency, enabling the reduction of the up-to-$t$ PID model to the exact-$t$ model.

\begin{lemma}[Transitivity of Standardization]\label{lem:std-transitive}
Let $x=(x_1,\ldots,x_m)$ have pairwise distinct entries and let $J=\{j_1<\cdots<j_s\}\subseteq[m]$. Then
\begin{equation*}
    \std\left(\std(x)_{j_1},\ldots,\std(x)_{j_s}\right)=\std(x_{j_1},\ldots,x_{j_s}).
\end{equation*}
\end{lemma}

\begin{proof}
For any $i, j \in J$, the relation $\std(x)_i < \std(x)_j$ holds if and only if $x_i < x_j$. Consequently, both subsequences induce the same strict total order and therefore have identical standardizations.
\end{proof}

For $\pi,\sigma\in\Sn_m$, let $\LCS(\pi,\sigma)$ denote the length of a longest common subsequence of $\pi$ and $\sigma$. Its order-isomorphic counterpart is defined as
\begin{equation*}
  \LPat(\pi,\sigma) \triangleq \max\left\{s : I,J\subseteq[m],\, |I|=|J|=s,\, \std(\pi_I)=\std(\sigma_J)\right\}.
\end{equation*}
The Ulam distance and the pattern distance between $\pi$ and $\sigma$ are defined respectively by
\begin{equation*}
  d_{\mathrm U}(\pi,\sigma) \triangleq m-\LCS(\pi,\sigma), \quad d_{\mathrm P}(\pi,\sigma) \triangleq m-\LPat(\pi,\sigma).
\end{equation*}
Two permutations share a common exact-$t$ SID descendant if and only if $d_{\mathrm U}(\pi,\sigma)\le t$, whereas they share a common exact-$t$ PID descendant if and only if $d_{\mathrm P}(\pi,\sigma)\le t$.

\begin{example}
Let $\pi=(1,3,2)$ and $\sigma=(2,3,1)$. Setting $I=J=\{1,2\}$, we obtain
\begin{equation*}
\std(\pi_I)=\std(1,3)=(1,2)=\std(2,3)=\std(\sigma_J),
\end{equation*}
which implies $\LPat(\pi,\sigma)=2$ and $\LCS(\pi,\sigma)=1$. Thus, $d_{\mathrm U}(\pi,\sigma)=2$ and $d_{\mathrm P}(\pi,\sigma)=1$. Consequently, $\pi$ and $\sigma$ are confusable under a single PID, but not under a single SID. This shows that $\LPat(\pi,\sigma)$ only requires the selected subsequences to have the same relative order, whereas $\LCS(\pi,\sigma)$ requires identical values.
\end{example}

Both deletion operations compose. For the SID channel this holds trivially, while for the PID channel it follows from Lemma~\ref{lem:std-transitive}. Specifically, if $\tau\in \mathcal D_s^{\mathsf x}(\pi)\cap \mathcal D_s^{\mathsf x}(\sigma)$ for some $s < t$, deleting any additional $t-s$ symbols from $\tau$ produces a descendant $\rho \in \mathcal D_t^{\mathsf x}(\pi) \cap \mathcal D_t^{\mathsf x}(\sigma)$. Therefore, for $m > t$, two permutations share a common descendant after at most $t$ deletions if and only if they share a common descendant after exactly $t$ deletions.

To analyze the independence number of the exact-$t$ confusability graph, we bound the number of parent permutations that can produce a given deletion descendant.

\begin{lemma}[Parent Enumeration]\label{lem:parent-enumeration}
Let $t\ge0$.
\begin{enumerate}
  \item Every sequence $x\in\mathcal I_{m+t,m}$ has exactly $t!\binom{m+t}{t}$ SID parents in $\Sn_{m+t}$.
  \item Every permutation $\tau\in\Sn_m$ has at most $t!\binom{m+t}{t}^2$ PID parents in $\Sn_{m+t}$.
\end{enumerate}
Furthermore, all such parents can be enumerated in $(m+t)^{O(t+1)}$ time.
\end{lemma}

\begin{proof}
For the SID channel, the $t$ missing symbols are uniquely identified by $[m+t] \setminus \{x_1, \ldots, x_m\}$. There are $\binom{m+t}{t}$ ways to choose their positions in $\Sn_{m+t}$ and $t!$ ways to arrange them. For the PID channel, a parent permutation is determined by choosing $t$ positions for the inserted symbols, choosing $t$ values from $[m+t]$ for these symbols, and ordering them in $t!$ ways. The remaining $m$ values are assigned by applying the unique order-preserving bijection from $[m]$ to the complement of the inserted value set to the symbols of $\tau$. These combinatorial choices yield the claimed counts and the corresponding enumeration algorithm.
\end{proof}

\subsection{Permutation Codes and Systematicity}

\begin{definition}[Deletion-Correcting Permutation Code]
For $\mathsf x\in\{\SID,\PID\}$, a code $\mathcal C\subseteq\Sn_m$ corrects up to $t$ $\mathsf x$-deletions if
\begin{equation*}
    \mathcal D_{\le t}^{\mathsf x}(\pi)\cap\mathcal D_{\le t}^{\mathsf x}(\sigma)=\varnothing
\end{equation*}
for all distinct $\pi,\sigma\in\mathcal C$. The redundancy of $\mathcal C$ is defined as $\red(\mathcal C)\triangleq\log\bigl(m!/|\mathcal C|\bigr)$.
\end{definition}

Let $G^{\mathsf x}_{m,t}$ denote the exact-$t$ confusability graph on the vertex set $\Sn_m$, where distinct permutations are adjacent if and only if their exact-$t$ $\mathsf x$-descendant sets intersect. By the exact-$t$ reduction property, up-to-$t$-deletion-correcting codes correspond to independent sets in $G^{\mathsf x}_{m,t}$. Let $\alpha(G)$ denote the independence number of a graph $G$, and let $M_{\mathsf x}(m,t)$ be the maximum cardinality of a $t$-$\mathsf x$-deletion-correcting code in $\Sn_m$. Then,
\begin{equation*}
  M_{\mathsf x}(m,t) = \alpha(G^{\mathsf x}_{m,t}),
\end{equation*}
and the optimal redundancy is defined by
\begin{equation*}
  R_{\mathsf x}(m,t) \triangleq \log\frac{m!}{M_{\mathsf x}(m,t)}.
\end{equation*}

\begin{definition}[Full-Systematic Encoder]\label{def:full-systematic}
Let $r\ge0$ be an integer and define $\mathcal R_{n,r} \triangleq [n+1,n+r]$, with the convention $\mathcal R_{n,0} \triangleq \varnothing$. An injective map $\Enc_n$ from $\Sn_n$ to $\Sn_{n+r}$ is full-systematic if
\begin{equation*}
  \Enc_n(\mu)|_{[n]} = \mu
\end{equation*}
for every message permutation $\mu\in\Sn_n$. The set $\mathcal R_{n,r}$ comprises the redundancy symbols. Since $|\Enc_n(\Sn_n)| = n!$, the redundancy of $\Enc_n$ is given by
\begin{equation*}
  \red(\Enc_n) \triangleq \log \frac{(n+r)!}{n!} = \sum_{j=1}^r\log(n+j).
\end{equation*}
\end{definition}

\subsection{Bounds and Benchmarks}

The asymptotic bounds on the maximum code sizes for both deletion models are summarized in the following theorem. The lower bound for the SID model is the improved Gilbert-Varshamov bound established in~\cite[Theorem~1]{wang2025permutation}. We establish the PID counterpart in Appendix~\ref{app:proof-thm-bounds} by analyzing the independence number via triangle enumeration in the confusability graph.

\begin{theorem}[Bounds for SID and PID Codes]\label{thm:prelim-bounds}
For every fixed integer $t\ge1$ and all sufficiently large $m$,
\begin{align*}
  \Omega_t\left(\frac{m!\log m}{m^{2t}}\right) &\le M_{\SID}(m,t) \le O_t\left(\frac{m!}{m^{t}}\right),
  \\
  \Omega_t\left(\frac{m!\log m}{m^{3t}}\right) &\le M_{\PID}(m,t) \le O_t\left(\frac{m!}{m^{t}}\right).
\end{align*}
\end{theorem}

The difference between the degree exponents in the lower bounds originates from the structural information lost under the two channels. In particular, a given SID descendant has $O_t(m^t)$ potential parents, whereas a PID descendant has $O_t(m^{2t})$ parents because the absolute values of the deleted symbols cannot be identified from the standardized sequence.

%% file: main_paper/pid.tex
\section{Systematic Codes For Permutation-Invariant Deletions}\label{sec:pid}

In this section, we present a construction of full-systematic permutation codes capable of correcting a fixed number of PIDs. Recall that, unlike the SID model studied in~\cite{wang2025permutation}, a PID standardizes the surviving symbols after each deletion, thereby obscuring the absolute values of the deleted elements. Our construction relies on a two-stage framework. First, we design a polynomial-time computable inner syndrome map whose syndrome classes each correct up to $t$ PIDs, where the total number of syndromes is bounded by $O_t(n^{6t-2})$. Second, we introduce an algebraic outer code that embeds this syndrome into the relative positions of $7t-1$ distinct redundancy markers. In the outer stage, $t$ integer moments narrow the parent candidates down to a list of constant size, after which a greedy coloring scheme guarantees unambiguous decoding. Consequently, the resulting code achieves a redundancy of $(7t-1)\log n+O_t(1)$ bits.

\subsection{A Polynomial-Time Inner PID Syndrome}
\label{subsec:inner-pid-syndrome}

For a permutation $\pi=(\pi_1,\ldots,\pi_n)\in\mathcal S_n$, its binary signature $\alpha(\pi)\in\{0,1\}^{n-1}$ is defined by
\begin{equation*}
   \alpha(\pi)_i \triangleq \mathbf 1\{\pi_{i+1}>\pi_i\},
   \qquad i\in[n-1].
\end{equation*}
We denote by $\mathbb Z_n$ the ring of integers modulo $n$. For a binary vector $x\in\{0,1\}^{n-1}$, we define
\begin{equation*}
   \VT_n(x) \triangleq \sum_{i=1}^{n-1} i x_i\pmod n,
\end{equation*}
and the signature syndrome
\begin{equation}
 \nu_n(\pi) \triangleq
 \bigl(\VT_n(\alpha(\pi)),\VT_n(\alpha(\pi^{-1}))\bigr)
 \in\mathbb Z_n^2.
 \label{eq:single-pid-syndrome}
\end{equation}
The following lemma recalls the single-PID-correcting code established in~\cite{gabrys2015codes}.

\begin{lemma}[Single-PID Correcting Codes~\cite{gabrys2015codes}]
\label{lem:single-pid-partition}
For every $v\in\mathbb Z_n^2$, the syndrome class $\nu_n^{-1}(v)$ corrects a single PID. In particular, for every $\rho\in\mathcal S_{n-1}$ and every $v\in\mathbb Z_n^2$,
\begin{equation*}
 \bigl|\{\pi\in\mathcal S_n :
       \rho\in\cD_1^{\PID}(\pi),\ \nu_n(\pi)=v\}\bigr|\leq1.
\end{equation*}
Moreover, $\nu_n(\pi)$ and the corresponding single-PID decoding algorithm are computable in polynomial time.
\end{lemma}

For $v\in\mathbb Z_n^2$, let $G_{n,t}(v)$ denote the residual graph on the vertex set $\nu_n^{-1}(v)$, where distinct permutations $\pi$ and $\sigma$ are adjacent if and only if $\cD_t^{\PID}(\pi)\cap\cD_t^{\PID}(\sigma)\neq\varnothing$. To bound the maximum degree of $G_{n,t}(v)$, we first bound the number of compatible parents of a fixed exact-$t$ descendant within the syndrome class $\nu_n^{-1}(v)$.

\begin{lemma}[Residual Parent List]
\label{lem:residual-parent-list}
Let $\tau\in\mathcal S_{n-t}$ and $v\in\mathbb Z_n^2$. Then
\begin{equation*}
 \bigl|\{\pi\in\mathcal S_n : \tau\in\cD_t^{\PID}(\pi),\ \nu_n(\pi)=v\}\bigr|\leq L_{n,t},
\end{equation*}
where
\begin{equation}
 L_{n,t} \triangleq \binom{n-1}{t-1}^{\!2}(t-1)!.
 \label{eq:Lnt}
\end{equation}
\end{lemma}

\begin{proof}
For every parent $\pi$ in the considered set, there exists an intermediate permutation $\rho\in\mathcal S_{n-1}$ such that $\rho\in\cD_1^{\PID}(\pi)$ and $\tau\in\cD_{t-1}^{\PID}(\rho)$. By Lemma~\ref{lem:parent-enumeration}, the number of candidate permutations $\rho$ is at most $L_{n,t}$. For each fixed $\rho$, Lemma~\ref{lem:single-pid-partition} implies that at most one parent $\pi$ can have syndrome $v$. Taking the union over all possible intermediate permutations $\rho$ completes the proof.
\end{proof}

Summing the bound in Lemma~\ref{lem:residual-parent-list} over the exact-$t$ descendants of a given permutation yields the following degree bound for the residual graph.

\begin{corollary}[Residual Degree] \label{cor:residual-degree}
Every graph $G_{n,t}(v)$ has maximum degree at most
\begin{equation}
K_{n,t} \triangleq \binom nt\binom{n-1}{t-1}^{\!2}(t-1)!\leq (t-1)!\,n^{3t-2}.\label{eq:Knt}
\end{equation}
Furthermore, its neighborhood relation can be enumerated in $n^{O(t)}$ time.
\end{corollary}

\begin{proof}
A fixed permutation $\pi$ has at most $\binom nt$ exact-$t$ descendants. For each descendant, Lemma~\ref{lem:residual-parent-list} guarantees at most $L_{n,t}$ parents within the same $\nu_n$-syndrome class. A union bound shows that the maximum degree of $G_{n,t}(v)$ cannot exceed $K_{n,t}$. The inequality follows from $\binom nt\leq n^t$ and $\binom{n-1}{t-1}\leq n^{t-1}$.

For the enumeration algorithm, we list all exact-$t$ descendants of $\pi$, enumerate their $t$-PID parents by Lemma~\ref{lem:parent-enumeration}, and retain those that match the prescribed syndrome $\nu_n$. The total running time is bounded by $n^{O(t)}$.
\end{proof}

By the construction of $G_{n,t}(v)$, every color class in a proper vertex coloring contains no two permutations sharing an exact-$t$ PID descendant. Hence, by Lemma~\ref{lem:std-transitive}, each color class corrects up to $t$ PIDs. Constructing the second-stage syndrome therefore reduces to properly coloring each residual graph with an alphabet of small cardinality that is efficiently computable. While the Lehmer rank yields a trivial proper coloring, its palette has size $n!$. Fortunately, Corollary~\ref{cor:residual-degree} provides bounded and efficiently enumerable neighborhoods, which allows this large coloring to be compressed locally. The following lemma adapts the polynomial cover-free coloring technique from~\cite[Lemmas 13 and 15]{li2025distributed} to our setting.

\begin{lemma}[Local Algebraic Fingerprint]
\label{lem:local-fingerprint}
Let $G=(V,E)$ be a graph with maximum degree at most $K$, and let $c$ be a proper coloring from $V$ to $\{0,\ldots,U-1\}$. Let $q$ be a prime number and $\ell \triangleq \lceil\log_q U\rceil$. Assume $q>K(\ell-1)$. Let the base-$q$ expansion of $c(x)$ be written as $c(x)=\sum_{i=0}^{\ell-1}c_i(x)q^i$, and define the polynomial
\begin{equation*}
    P_x(X) \triangleq \sum_{i=0}^{\ell-1}c_i(x)X^i\in\mathbb F_q[X].
\end{equation*}
For each $x\in V$, let $N_G(x) \triangleq \{y\in V : \{x,y\}\in E\}$ denote the neighborhood of $x$ in $G$. Let $a(x)$ be the smallest element in $\mathbb F_q$ under the standard representative ordering such that $P_x(a(x))\neq P_y(a(x))$ for every $y\in N_G(x)$. Then $a(x)$ exists, and the mapping
\begin{equation*}
    \Phi_c(x) \triangleq (a(x),P_x(a(x)))\in\mathbb F_q^2
\end{equation*}
is a proper coloring. Given the values $c(y)$ for all $y\in\{x\}\cup N_G(x)$, the color $\Phi_c(x)$ can be computed in time polynomial in $q$, $\ell$, and $K$.
\end{lemma}

\begin{proof}
For any neighbor $y$ of $x$, the properness of $c$ ensures that $P_x\neq P_y$. Thus, $P_x-P_y$ is a nonzero polynomial of degree at most $\ell-1$, which has at most $\ell-1$ roots in $\mathbb F_q$. Since $|N_G(x)|\le K$, the union of the root sets of all such differences has cardinality at most $K(\ell-1)$. The condition $q>K(\ell-1)$ implies that at least one element of $\mathbb F_q$ is not a root, which proves the existence of $a(x)$.

If $\{x,y\}\in E$ and $\Phi_c(x)=\Phi_c(y)$, then their first coordinates must be equal, say to $a$. Since $a=a(x)$ and $y\in N_G(x)$, the definition of $a(x)$ enforces $P_x(a)\neq P_y(a)$, which contradicts the equality of their second coordinates. Hence, $\Phi_c$ is proper. The algorithm computes $\Phi_c(x)$ by sequentially evaluating the $q$ field elements in their fixed order.
\end{proof}

We apply Lemma~\ref{lem:local-fingerprint} twice. Assume $n\geq t+1$ and set $K \triangleq K_{n,t}$. Let $q_0$ and $q_1$ be the smallest primes satisfying $q_0>2nK$ and $q_1>4K$. By Bertrand's postulate, $q_0<4nK$ and $q_1<8K$. For a fixed syndrome $v\in\mathbb Z_n^2$, let $c_{0,v}(\pi)$ be the Lehmer rank of $\pi$, viewed as an integer in $\{0,\ldots,n!-1\}$. This provides an injective, and therefore proper, coloring of $G_{n,t}(v)$. Applying Lemma~\ref{lem:local-fingerprint} with $q=q_0$ yields a proper color $c_{1,v}(\pi)\in\mathbb F_{q_0}^2$. We then represent this pair as an integer in $\{0,\ldots,q_0^2-1\}$ and apply the lemma a second time with $q=q_1$, obtaining $c_{2,v}(\pi)\in\mathbb F_{q_1}^2$.

Both applications of the lemma satisfy the degree constraints. Since $q_0>n$, we have
\begin{equation*}
 \ell_0 \triangleq \lceil\log_{q_0}(n!)\rceil\leq n,\qquad K(\ell_0-1)<nK<q_0.
\end{equation*}
Furthermore, $K\geq n$ holds for all $n\geq t+1$. This is obvious when $t=1$. For $t\geq2$, one has $\binom{n-1}{t-1}\geq\binom{t}{t-1}=t$, and hence $\binom nt=\frac nt\binom{n-1}{t-1}\geq n$. Consequently,
\begin{equation*}
 q_0^2<(4nK)^2\leq16K^4<(4K)^4<q_1^4.
\end{equation*}
This yields $\ell_1 \triangleq \lceil\log_{q_1}(q_0^2)\rceil\leq4$ and $K(\ell_1-1)\leq3K<q_1$.

These two compression stages establish the inner syndrome map and its decoding capability.

\begin{theorem}[Polynomial-Time Inner PID Syndrome]
\label{thm:inner-pid-syndrome}
For every fixed integer $t\geq1$ and $n\geq t+1$, the inner syndrome map $\kappa_{n,t}$ from $\mathcal S_n$ to $\mathbb Z_n^2\times\mathbb F_{q_1}^2$ is defined by
\begin{equation}
  \kappa_{n,t}(\pi) \triangleq
  \bigl(\nu_n(\pi),c_{2,\nu_n(\pi)}(\pi)\bigr).
  \label{eq:inner-syndrome}
\end{equation}
Every syndrome class of $\kappa_{n,t}$ corrects up to $t$ PIDs. The cardinality of the syndrome space satisfies
\begin{equation*}
  M_{n,t} \triangleq n^2q_1^2<64n^2K_{n,t}^2
  \leq C_t n^{6t-2},\qquad
  C_t \triangleq 64((t-1)!)^2.
\end{equation*}
The map $\kappa_{n,t}$, the associated decoding procedure, and all design parameters can be computed in $n^{O(t)}$ time.
\end{theorem}

\begin{proof}
For each syndrome vector $v$, each fingerprinting operation preserves properness, ensuring that $c_{2,v}$ is a proper coloring of $G_{n,t}(v)$. Consequently, two distinct permutations with the same value of $\kappa_{n,t}$ cannot share an exact-$t$ PID descendant. By Lemma~\ref{lem:std-transitive}, their up-to-$t$ PID descendant sets are also mutually disjoint.

The size bound follows from $q_1<8K_{n,t}$ and Corollary~\ref{cor:residual-degree}. Regarding computational complexity, evaluating the first fingerprint at $\pi$ requires only the residual neighbors of $\pi$. The second fingerprint requires the first fingerprints of these neighbors, which restricts the search to radius-two residual neighborhoods. By Corollary~\ref{cor:residual-degree}, these neighborhoods can be enumerated in $n^{O(t)}$ time. Computing Lehmer ranks, evaluating the single-PID syndrome, performing finite-field operations, and conducting deterministic primality tests are all polynomial-time tasks. Given a descendant sequence $x\in\mathcal S_{n-s}$ for $s\leq t$ along with its syndrome value, we enumerate all $s$-PID parents of $x$ using Lemma~\ref{lem:parent-enumeration}, compute $\kappa_{n,t}$ for each candidate, and select the unique matching permutation. The number of candidate parents is at most $\binom ns^2s!=n^{O(t)}$, and uniqueness is guaranteed by the code property.
\end{proof}

\subsection{An Algebraic Outer Code for PID}
\label{subsec:pid-outer}

Because any permutation contains each alphabet symbol exactly once, inner syndromes cannot be appended as traditional parity-check symbols. Instead, we embed the syndrome into the positions occupied by distinct redundancy markers.

\subsubsection{Canonical labeled-gap words}

Fix $r\geq2t$. Let $D$ denote a data symbol placeholder, and let $1,\ldots,r$ be distinct marker labels. For a vector $z=(z_1,\ldots,z_r)\in\{0,\ldots,n\}^r$, the word $W_n(z)$ is formed by iterating $a$ from $0$ to $n$. At each step $a$, we place all markers $j$ satisfying $z_j=a$ in increasing order of their labels. Whenever $a<n$, this is followed by a single copy of $D$. Under this convention, $z_j$ indicates the number of data placeholders that precede marker $j$.

For a message $\mu=(\mu_1,\ldots,\mu_n)\in\mathcal S_n$, let $\Int_z(\mu)\in\mathcal S_{n+r}$ be the permutation obtained by substituting the $i$-th copy of $D$ in $W_n(z)$ with $\mu_i$ and replacing marker $j$ with the symbol $n+j$. Deleting the fixed redundancy symbols $n+1,\ldots,n+r$ recovers $\mu$.

\begin{example}\label{ex:marker}
Let $t=2$, $n=4$, $r=4$, and $z=(0,2,2,4)\in\{0,\ldots,4\}^4$. Then $W_4(z)=(1,D,D,2,3,D,D,4)$. In this word, marker $1$ has no preceding data placeholder, markers $2$ and $3$ each have two preceding placeholders, and marker $4$ has four preceding placeholders. For $\mu=(3,1,4,2)\in\mathcal S_4$, replacing the four copies of $D$ in order by the elements of $\mu$, and each marker $j$ by $4+j$, yields $\mathrm{Int}_z(\mu) =(5,3,1,6,7,4,2,8)\in\mathcal S_8$. Deleting the marker symbols $5,6,7,8$ recovers $(3,1,4,2)=\mu$.
\end{example}

\subsubsection{The projected PID outer channel}

Let $N \triangleq n+r$. For $0\leq d\leq t$ and $y\in\mathcal S_{N-d}$, we define the coordinate projection
\begin{equation}
 \Theta_n(y)_i \triangleq
 \begin{cases}
 D, & y_i\leq n,\\
 y_i-n, & y_i>n.
 \end{cases}
 \label{eq:theta-projection}
\end{equation}
We now characterize the set of words that can arise through this projection. For a fixed vector $z$, let $\delta$ and $e$ be non-negative integers satisfying $\delta+e=d\leq t$, where $\delta$ represents the number of deleted data symbols and $e$ represents the number of deleted markers. Choose a subset $E\subseteq[r]$ of size $e$ to be deleted, along with any $\delta$ copies of $D$, from $W_n(z)$. Let the surviving markers be indexed as $[r]\setminus E=\{s_1<s_2<\cdots<s_{r-e}\}$. We convert the first $\delta$ markers $s_1,\ldots,s_\delta$ into copies of $D$, and relabel $s_{\delta+k}$ as $k$ for each $k\in[r-d]$. The assumption $r\geq2t$ guarantees $r-e\geq\delta$. We denote by $\cT_{\leq t}(z)$ the set of all words that can be produced in this manner.

\begin{lemma}[PID Projection]
\label{lem:pid-projection}
Let $\mu\in\mathcal S_n$, let $y$ be obtained from $\Int_z(\mu)$ after $d\leq t$ PIDs, and let $\delta$ and $e$ denote the numbers of deleted data and marker symbols, respectively. Then
\begin{equation*}
       \Theta_n(y)\in\cT_{\leq t}(z).
\end{equation*}
In particular, the projected outer sequence depends on $\mu$ solely through the locations of the deleted data positions, and is completely invariant to the values of the surviving data symbols.
\end{lemma}

\begin{proof}
Let $s_1<\cdots<s_{r-e}$ denote the labels of the surviving markers. All surviving data symbols have smaller values than all surviving markers. Following standardization, the surviving data symbols take values in $[n-\delta]$, while the marker $s_k$ is assigned the value $n-\delta+k$. Whenever $k\leq\delta$, this value is at most $n$, so the projection in \eqref{eq:theta-projection} converts the marker into $D$. When $k>\delta$, the value equals $n+(k-\delta)$, which projects to the label $k-\delta$. Deleting the corresponding data positions from the placeholders of $W_n(z)$ reproduces this exact operation.
\end{proof}

\begin{example}\label{ex:projection}
Continuing Example~\ref{ex:marker} with $t=2, n=r=4, z=(0,2,2,4)$, and $\mu=(3,1,4,2)$, we have $W_4(z)=(1,D,D,2,3,D,D,4)$ and $\Int_z(\mu)=(5,3,1,6,7,4,2,8)$. Suppose we delete the second data symbol $\mu_2=1$ and marker $2$, whose symbol value is $n+2=6$. Then $\delta=e=1$ and $d=2$. After standardization, the resulting sequence is $y=\std(5,3,7,4,2,8)=(4,2,5,3,1,6)$. Applying the projection operator yields
\begin{equation*}
 \Theta_4(y)=(D,D,1,D,D,2).
\end{equation*}

The identical sequence is obtained directly from $W_4(z)$. Deleting its second copy of $D$ and marker $2$ leaves $(1,D,3,D,D,4)$. The surviving markers are $1,3,4$. Since $\delta=1$, the first surviving marker becomes $D$, while markers $3$ and $4$ are relabeled as $1$ and $2$, respectively. This again results in $\Theta_4(y)=(D,D,1,D,D,2)\in\mathcal T_{\leq 2}(z)$. Thus, the projected output is determined entirely by $z$ and the positions of the deletions, independent of the actual data values.
\end{example}

\subsubsection{Moment checks and constant-list reduction}

For each $\ell\in\{0,\ldots,t-1\}$, we define the moments
\begin{equation*}
      h_\ell(z) \triangleq \sum_{j=1}^r j^\ell z_j,\qquad H_{r,t}(z) \triangleq (h_0(z),\ldots,h_{t-1}(z))
      \in\mathbb Z_{\geq0}^t.
\end{equation*}
For any moment vector $h\in\mathbb Z_{\geq0}^t$, let
\begin{equation*}
      \cZ_h(n,r,t) \triangleq \{z\in\{0,\ldots,n\}^r : H_{r,t}(z)=h\}.
\end{equation*}
Although standardization conceals the labels of certain surviving markers, the moment vector restricts the number of compatible parents to a constant that depends only on $r$ and $t$.

\begin{lemma}[Constant Parent List in a Moment Class]
\label{lem:outer-parent-list}
Fix $h$ and an outer word $w$ of length $n+r-d$ with $0\leq d\leq t$. Then
\begin{equation*}
    \bigl|\{z\in\cZ_h(n,r,t) : w\in\cT_{\leq t}(z)\}\bigr|\leq P_{r,t},
\end{equation*}
where
\begin{equation}
 P_{r,t} \triangleq \max_{0\leq d\leq t}
 \sum_{\delta=0}^{d}
 \binom r{d-\delta}(2\delta+1)^{r-d}.
 \label{eq:Prt}
\end{equation}
In particular, $P_{r,t}$ is strictly independent of $n$.
\end{lemma}

\begin{proof}
The length of the output word uniquely determines $d$. Fix a choice of $\delta$, let $e=d-\delta$, and fix the deleted marker set $E\subseteq[r]$ of size $e$. We partition the surviving markers as
\begin{equation*}
 [r]\setminus E=\{s_1<\cdots<s_{r-e}\},\quad L=\{s_1,\ldots,s_\delta\},\quad
 G=([r]\setminus E)\setminus L.
\end{equation*}
The markers in $L$ are precisely those transformed into $D$, while those in $G$ remain visible. For each $k\in[r-d]$, the visible output marker $k$ corresponds to the original marker $j=s_{\delta+k}$. Let $q_k$ denote the number of placeholders $D$ preceding marker $k$ in $w$.

Suppose $z$ is consistent with these choices. Let $a_j$ be the number of deleted data placeholders preceding marker $j$, and let $b_j$ be the number of markers in $L$ preceding $j$ in the sequence. Then $q_k=z_j-a_j+b_j$, where $0\leq a_j,b_j\leq\delta$. This implies $|z_j-q_k|\leq\delta$. Consequently, each coordinate $z_j$ for $j\in G$ admits at most $2\delta+1$ possibilities, which bounds the number of choices for $z_G$ by $(2\delta+1)^{r-d}$.

It remains to determine the coordinates indexed by $U \triangleq E\cup L$, where $|U|=d\leq t$. For any fixed choice of $z_G$, the moment constraint $H_{r,t}(z)=h$ gives
\begin{equation*}
  \sum_{j\in U}j^\ell z_j
  =h_\ell-\sum_{j\in G}j^\ell z_j,
  \qquad \ell=0,\ldots,t-1.
\end{equation*}
When $d=0$, all coordinates are known. When $d\geq1$, the submatrix formed by the first $d$ equations and the columns in $U$ constitutes a $d\times d$ Vandermonde matrix. Its determinant is nonzero because all indices in $U$ are distinct. Therefore, the coordinates $z_U$ are uniquely determined if a solution exists. Summing over the $\binom r{d-\delta}$ choices for $E$ and all valid values of $\delta$ yields the bound in \eqref{eq:Prt}.
\end{proof}

\begin{example}
We continue with $n=r=4, t=2$, and $z=(0,2,2,4)$. The moment vector is $H_{4,2}(z)=(8,26)$, and the deletion pattern $\delta=e=1$ with $E=\{2\}$ yields the output $w=(D,D,1,D,D,2)\in\mathcal T_{\leq2}(z)$.

In this configuration, we have $L=\{1\}, G=\{3,4\}$, and $U=\{1,2\}$. The two visible markers are preceded by $q_1=2$ and $q_2=4$ copies of $D$, so $|z_3-2|\leq1$ and $|z_4-4|\leq1$. Fixing $z_3$ and $z_4$ leads to the linear system
\begin{equation*}
\begin{pmatrix}
1&1\\
1&2
\end{pmatrix}
\begin{pmatrix}
z_1\\z_2
\end{pmatrix}
=
\begin{pmatrix}
8-z_3-z_4\\
26-3z_3-4z_4
\end{pmatrix}.
\end{equation*}
Because the matrix is nonsingular, each selection of the visible coordinates has at most one completion. Setting $(z_3,z_4)=(2,4)$ uniquely yields $(z_1,z_2)=(0,2)$.

Note that completions need not be unique across different deletion patterns. For instance, the vector $\tilde z=(0,1,4,3)$ also satisfies $H_{4,2}(\tilde z)=(8,26)$, and deleting its fourth copy of $D$ along with marker $3$ generates the identical word $w$. The moment checks reduce candidates to a list of constant size, and the subsequent graph coloring guarantees unique identification.
\end{example}

\begin{lemma}[Output Bound]
\label{lem:outer-output-count}
For every vector $z\in\{0,\ldots,n\}^r$,
\begin{equation}
 |\cT_{\leq t}(z)|\leq 
 \sum_{d=0}^{t}\sum_{\delta=0}^{d}
 \binom r{d-\delta}\binom{\delta+r}{r} \triangleq A_{r,t}.
 \label{eq:Art}
\end{equation}
\end{lemma}

\begin{proof}
Fix $d$, $\delta$, and a deleted marker set $E$. The placeholders $D$ in $W_n(z)$ form at most $r+1$ runs. Once the deleted markers are specified, the resulting word is determined entirely by how the $\delta$ deletions of $D$ are distributed among these runs. The number of such distributions corresponds to the number of weak compositions, which is at most $\binom{\delta+r}{r}$. All subsequent masking and relabeling steps are deterministic. Summing over all valid parameters establishes \eqref{eq:Art}.
\end{proof}

For each moment vector $h$, define the residual outer graph $\cG_h(n,r,t)$ on the vertex set $\cZ_h(n,r,t)$, where distinct vectors $z,z'$ are adjacent if and only if $\cT_{\leq t}(z)\cap\cT_{\leq t}(z')\neq\varnothing$. Lemmas~\ref{lem:outer-parent-list} and~\ref{lem:outer-output-count} imply that the maximum degree satisfies
\begin{equation}
  \Delta(\cG_h(n,r,t)) \leq A_{r,t}(P_{r,t}-1).
 \label{eq:outer-degree}
\end{equation}
Let $B_{r,t} \triangleq A_{r,t}(P_{r,t}-1)+1$. We order $\cZ_h(n,r,t)$ lexicographically and greedily assign each vertex the smallest color in $[B_{r,t}]$ that does not appear among its previously colored neighbors. We denote this canonical coloring by $\chi_h$. The complete outer index map is then defined as
\begin{equation}
   \Omega_{n,r,t}(z) \triangleq
   \bigl(H_{r,t}(z),\chi_{H_{r,t}(z)}(z)\bigr).
 \label{eq:outer-index}
\end{equation}

\begin{theorem}[Algebraic PID Outer Partition]
\label{thm:pid-outer-partition}
Let $r\geq2t$. Every syndrome class of $\Omega_{n,r,t}$ corrects the projected outer channel $\cT_{\leq t}$. Furthermore, if
\begin{equation*}
 S_\ell(r) \triangleq \sum_{j=1}^r j^\ell,
 \qquad
 Q_{r,t} \triangleq \prod_{\ell=0}^{t-1}(S_\ell(r)+1),
\end{equation*}
then $|\operatorname{im}\Omega_{n,r,t}|\leq B_{r,t}Q_{r,t}(n+1)^t$. Consequently, there exists a syndrome class $\cA_{n,r,t}$ with size satisfying
\begin{equation}
 |\cA_{n,r,t}|
 \geq\frac{(n+1)^{r-t}}{B_{r,t}Q_{r,t}}
 =\Omega_{r,t}(n^{r-t}).
 \label{eq:outer-syndrome class-size}
\end{equation}
For fixed parameters $r$ and $t$, the partition, the lexicographically first largest syndrome class, and the ranking and unranking algorithms within that class can all be computed in time polynomial in $n$.
\end{theorem}

\begin{proof}
If two distinct vectors $z,z'$ have identical values of $\Omega_{n,r,t}$, they share the same moment vector and receive the same proper color in the residual graph. Thus, they cannot be adjacent in $\cG_{H_{r,t}(z)}(n,r,t)$, which guarantees that their outer output sets are disjoint. Each moment coordinate $h_\ell(z)$ ranges from $0$ to $nS_\ell(r)$, which implies
\begin{equation*}
 |\operatorname{im}(H_{r,t})|
 \leq\prod_{\ell=0}^{t-1}(nS_\ell(r)+1)
 \leq Q_{r,t}(n+1)^t.
\end{equation*}
Multiplying by the $B_{r,t}$ residual colors gives the bound on $|\operatorname{im}(\Omega_{n,r,t})|$. Because the domain contains $(n+1)^r$ vectors, the pigeonhole principle guarantees that the largest syndrome class meets the average size bound in \eqref{eq:outer-syndrome class-size}.

For the algorithmic complexity, we generate all $(n+1)^r$ vectors in lexicographic order. For each vector, we compute its moment vector and enumerate its at most $A_{r,t}$ outer outputs. By hashing each output together with its moment vector, all residual neighbors are identified, after which the canonical colors are assigned. We then select the lexicographically first syndrome value that achieves the maximal class size. Because $r$ and $t$ are constants, the total time and storage complexities are polynomial in $n$.
\end{proof}

The following proposition provides a matching upper bound, showing that the scaling order $\Omega_{r,t}(n^{r-t})$ in \eqref{eq:outer-syndrome class-size} is order-optimal.

\begin{proposition}[Optimal Outer Exponent]
\label{prop:optimal-outer-exponent}
Let $\cA\subseteq\{0,\ldots,n\}^r$ be a code such that the output sets $\cT_{\leq t}(z)$ are mutually disjoint for all $z\in\cA$. Then
\begin{equation*}
       |\cA|\leq\frac{(n+r-t)!}{n!}=O_{r,t}(n^{r-t}).
\end{equation*}
\end{proposition}

\begin{proof}
Consider the deletion pattern that removes the first $t$ markers $\{1,\ldots,t\}$ and no data placeholders from each $W_n(z)$. The resulting word contains $n$ copies of $D$ and $r-t$ distinct relabeled markers. This mapping must be injective on $\cA$, and the total number of such sequences is $(n+r-t)!/n!$.
\end{proof}

\subsection{The Full-Systematic Construction}
\label{subsec:full-systematic-composition}

Set $r \triangleq 7t-1$, $B_t \triangleq B_{7t-1,t}$, and $Q_t \triangleq Q_{7t-1,t}$. Let $\cA_{n,t} \triangleq \cA_{n,7t-1,t}$ denote the lexicographically first largest syndrome class obtained from Theorem~\ref{thm:pid-outer-partition}. By \eqref{eq:outer-syndrome class-size}, we have $|\cA_{n,t}|\geq\frac{(n+1)^{6t-1}}{B_tQ_t}$. In addition, Theorem~\ref{thm:inner-pid-syndrome} guarantees $M_{n,t}\leq C_tn^{6t-2}$. Hence, for all
\begin{equation}
 n\geq n_0(t) \triangleq
 \max\{t+1,\lceil C_tB_tQ_t\rceil\},
 \label{eq:n0}
\end{equation}
we have the inequality $|\cA_{n,t}|\geq M_{n,t}$. We express the inner syndrome as
\begin{equation*}
 \kappa_{n,t}(\mu)=(a_1,a_2,b_1,b_2)\in\mathbb Z_n^2\times\mathbb F_{q_1}^2
\end{equation*}
using the standard integer representatives of the field elements. This tuple is mapped to the integer index
\begin{equation}
 \operatorname{idx}(a_1,a_2,b_1,b_2) \triangleq ((a_1n+a_2)q_1+b_1)q_1+b_2
 \in\{0,\ldots,M_{n,t}-1\}.\label{eq:mixed-radix-index}
\end{equation}
Let $z^{(q)}$ denote the vector of rank $q$ in the lexicographically ordered syndrome class $\cA_{n,t}$.

\begin{construction}[Systematic PID Encoder]
\label{con:full-systematic-pid-encoder}
Given a message permutation $\mu\in\mathcal S_n$, the encoding procedure is executed in three steps.
\begin{enumerate}[label=\arabic*)]
\item Compute the inner syndrome $\kappa_{n,t}(\mu)$ and its integer index $q$ via \eqref{eq:mixed-radix-index}.
\item Determine the vector $z=z^{(q)}\in\cA_{n,t}$ by unranking the index $q$.
\item Output the codeword $\Enc_{n,t}^{\PID}(\mu) \triangleq \Int_z(\mu)\in\mathcal S_{n+7t-1}$.
\end{enumerate}
\end{construction}

The encoder is full-systematic by definition of $\Int_z$. The next lemma shows that a common sequence output implies a common pattern in the underlying message permutations.

\begin{lemma}[Common PID Output Induces Common Subsequence Pattern]
\label{lem:common-data-pattern}
Suppose that $\Int_z(\mu)$ and $\Int_{z'}(\nu)$ produce a common descendant $y\in\mathcal S_{n+r-d}$ after $d\leq t$ PIDs. Let $\delta_\mu$ and $\delta_\nu$ denote the numbers of deleted data symbols in each codeword, respectively. Then $\mu$ and $\nu$ share a common PID descendant of length $n-\max\{\delta_\mu,\delta_\nu\}\geq n-t$.
\end{lemma}

\begin{proof}
Without loss of generality, assume $\delta_\mu\geq\delta_\nu$, and define $a \triangleq n-\delta_\mu$ and $b \triangleq n-\delta_\nu$. In the standardized sequence obtained from the first codeword, the surviving data symbols correspond to the values in $[a]$. In the sequence from the second codeword, they correspond to the values in $[b]$, with $a\leq b$. Restricting $y$ to the entries with values in $[a]$ yields the standardized sequence of all surviving data symbols for the first word, and represents the deletion of an additional $b-a$ symbols from the second. By Lemma~\ref{lem:std-transitive}, this sequence in $\mathcal S_a$ is a PID descendant of both $\mu$ and $\nu$. The total number of data deletions is $n-a=\delta_\mu\leq d\leq t$.
\end{proof}

We now state the decoding procedure. During the construction of $\cA_{n,t}$, we maintain a dictionary that maps each sequence $w\in\bigcup_{z\in\cA_{n,t}}\cT_{\leq t}(z)$ to its unique parent vector $z$.

\begin{construction}[Systematic PID Decoder]
\label{con:full-systematic-pid-decoder}
Given a received sequence $y\in\mathcal S_{n+r-d}$ produced by at most $d\leq t$ PIDs, decoding is performed in three steps.
\begin{enumerate}[label=\arabic*)]
\item Compute $w=\Theta_n(y)$ and retrieve the unique vector $z\in\cA_{n,t}$ satisfying $w\in\cT_{\leq t}(z)$ from the outer dictionary.
\item Determine the rank $q$ of $z$ in $\cA_{n,t}$. If $q\geq M_{n,t}$, declare a decoding failure. Otherwise, invert \eqref{eq:mixed-radix-index} to recover the inner syndrome $\gamma\in\mathbb Z_n^2\times\mathbb F_{q_1}^2$.
\item For each $\delta\in\{0,\ldots,d\}$, extract from $y$ the subsequence $x_\delta\in\mathcal S_{n-\delta}$ consisting of all symbols with values at most $n-\delta$. Enumerate all $\delta$-PID parents $\mu'\in\mathcal S_n$ of $x_\delta$. Finally, output the unique candidate satisfying $\kappa_{n,t}(\mu')=\gamma$ and $y\in\cD_d^{\PID}(\Int_z(\mu'))$.
\end{enumerate}
\end{construction}

The final consistency check is implemented by inspecting the $\binom{n+r}{d}$ deletion choices. Thus, the decoder runs in $n^{O(t)}$ time.

\begin{theorem}[Full-Systematic $t$-PID Correcting Code]
\label{thm:main-systematic-pid}
For every fixed integer $t\geq1$ and all $n\geq n_0(t)$, the encoder $\Enc_{n,t}^{\PID}$ from $\mathcal S_n$ to $\mathcal S_{n+7t-1}$ specified in Construction~\ref{con:full-systematic-pid-encoder} is full-systematic, and its codebook corrects up to $t$ PIDs. The code family is uniformly constructible, encodable, and decodable in $n^{O(t)}$ time. The redundancy satisfies
\begin{equation*}
 \red(\Enc_{n,t}^{\PID}) = (7t-1)\log n+O_t(1).
\end{equation*}
\end{theorem}

\begin{proof}
Suppose two encoded permutations $\Enc_{n,t}^{\PID}(\mu)$ and $\Enc_{n,t}^{\PID}(\nu)$ share an output $y$ after $d\leq t$ PIDs. By Lemma~\ref{lem:pid-projection}, the projected word $w=\Theta_n(y)$ lies in the outer output sets of both vectors $z$ and $z'$. Because both vectors belong to the same outer class $\cA_{n,t}$, Theorem~\ref{thm:pid-outer-partition} enforces $z=z'$. Their ranks must therefore coincide, and the injectivity of the indexing in \eqref{eq:mixed-radix-index} implies $\kappa_{n,t}(\mu)=\kappa_{n,t}(\nu)$. According to Lemma~\ref{lem:common-data-pattern}, $\mu$ and $\nu$ share a common PID descendant after at most $t$ deletions. Theorem~\ref{thm:inner-pid-syndrome} then establishes that $\mu=\nu$.

This argument also ensures the uniqueness of the decoded message in Construction~\ref{con:full-systematic-pid-decoder}. The inner syndrome evaluation and decoding run in polynomial time by Theorem~\ref{thm:inner-pid-syndrome}. The outer partition, dictionary lookup, and rank conversions operate in polynomial time by Theorem~\ref{thm:pid-outer-partition}. All remaining parent enumerations require at most $n^{O(t)}$ operations, which proves the complexity bound.

Finally, setting $r=7t-1$, the code redundancy evaluates to
\begin{align*}
 \red(\Enc_{n,t}^{\PID})
 &= \log\frac{(n+r)!}{n!}
   = \sum_{j=1}^{r}\log(n+j)
   = (7t-1)\log n+O_t(1).
\end{align*}\qedhere
\end{proof}

%% file: main_paper/sid.tex
\section{Systematic Codes For Symbol-Invariant Deletions}\label{sec:sid}

We next construct full-systematic permutation codes capable of correcting up to $t$ SIDs. The labeled-marker outer coding framework developed in Section~\ref{subsec:pid-outer} applies directly to this setting, as an SID output retains the absolute labels of the surviving markers and can therefore be mapped deterministically to a sequence in the projected outer channel. The primary technical distinction lies in the design of the inner syndrome. By leveraging the cyclic-successor mapping introduced in~\cite{wang2025permutation}, we construct an inner SID syndrome that takes at most $O_t(n^{3t-1})$ possible values. Storing this syndrome in an outer code class with $4t$ redundancy markers achieves an overall redundancy of $4t\log n+O_t(1)$ bits.

\subsection{A Polynomial-Time Inner SID Syndrome}

Let $N \triangleq n+1$, let $p$ be the smallest prime satisfying $p\geq N$, and let $\iota$ denote the canonical injection from $[N]$ to $\mathbb F_p$ defined by $\iota(a) \triangleq a\pmod p$. When $N=p$, the element $p$ is identified with $0$, so that all $N$ images remain pairwise distinct. For an integer $d$ with $2\leq d\leq N$ and a permutation $\rho=(\rho_1,\ldots,\rho_N)\in\mathcal S_N$, we define the power-sum syndrome
\begin{equation}
 J_{N,d}(\rho) \triangleq \left(\sum_{i=1}^{N}\iota(i)^k\iota(\rho_i)\right)_{k=1}^{d-2}
 \in\mathbb F_p^{d-2}.
 \label{eq:sid-power-sum-index}
\end{equation}
Thus, $J_{N,d}(\rho)$ records $d-2$ weighted power sums of the entries of $\rho$; when $d=2$, it represents the empty tuple.

\begin{lemma}[Canonical Hamming Partition]
\label{lem:canonical-hamming-partition}
For every vector $\bm a\in\mathbb F_p^{d-2}$, the set $\bigl\{\rho\in\mathcal S_N : J_{N,d}(\rho)=\bm a\bigr\}$ has minimum Hamming distance at least $d$. Consequently, $J_{N,d}$ partitions $\mathcal S_N$ into at most $p^{d-2}$ permutation codes of minimum Hamming distance at least $d$. Furthermore, $J_{N,d}(\rho)$ can be evaluated using $O(Nd)$ operations over $\mathbb F_p$.
\end{lemma}

\begin{proof}
Let $\rho,\sigma\in\mathcal S_N$ be two distinct permutations satisfying $J_{N,d}(\rho)=J_{N,d}(\sigma)$. Let $I \triangleq \{i\in[N] : \rho_i\neq\sigma_i\}=\{i_1,\ldots,i_h\}$, where $h=d_{\rm H}(\rho,\sigma)$. For each $j\in[h]$, define $c_j \triangleq \iota(\rho_{i_j})-\iota(\sigma_{i_j})\in\mathbb F_p$.

Because $\rho_{i_j}\neq\sigma_{i_j}$ and $\iota$ is injective, we have $c_j\neq 0$ for all $j\in[h]$. Since both $\rho$ and $\sigma$ are permutations of the set $[N]$, their coordinate sums satisfy $\sum_{i=1}^{N}\iota(\rho_i)=\sum_{i=1}^{N}\iota(\sigma_i)$. The two permutations coincide outside $I$, which implies $\sum_{j=1}^{h} c_j=0$. Moreover, the condition $J_{N,d}(\rho)=J_{N,d}(\sigma)$ enforces $\sum_{i=1}^{N}\iota(i)^k\bigl(\iota(\rho_i)-\iota(\sigma_i)\bigr)=0$ for every $k\in[d-2]$.

Using the fact that $\rho_i=\sigma_i$ for all $i\notin I$, we obtain
\begin{equation}
    \sum_{j=1}^{h}
        \iota(i_j)^k c_j=0,
    \qquad
    k=1,\ldots,d-2.
    \label{eq:sum0}
\end{equation}

Suppose, for contradiction, that $h\le d-1$. Then $h-1\le d-2$. The relation $\sum_{j=1}^{h} c_j=0$ together with the first $h-1$ equations in \eqref{eq:sum0} yields the linear system
\begin{equation}
\begin{pmatrix}
1
    & 1
    & \cdots
    & 1
\\
\iota(i_1)
    & \iota(i_2)
    & \cdots
    & \iota(i_h)
\\
\iota(i_1)^2
    & \iota(i_2)^2
    & \cdots
    & \iota(i_h)^2
\\
\vdots
    & \vdots
    & \ddots
    & \vdots
\\
\iota(i_1)^{h-1}
    & \iota(i_2)^{h-1}
    & \cdots
    & \iota(i_h)^{h-1}
\end{pmatrix}
\begin{pmatrix}
c_1\\
c_2\\
\vdots\\
c_h
\end{pmatrix}
=
\begin{pmatrix}
0\\
0\\
\vdots\\
0
\end{pmatrix}.
\label{eq:vanmagtrix}
\end{equation}
The coefficient matrix in \eqref{eq:vanmagtrix} is a Vandermonde matrix with determinant $\prod_{1\le u<v\le h}\bigl(\iota(i_v)-\iota(i_u)\bigr)$. Since the indices $i_1,\ldots,i_h$ are distinct and $\iota$ is injective, we have $\iota(i_u)\neq\iota(i_v)$ for all $u\neq v$, ensuring that the determinant is nonzero in $\mathbb F_p$. Hence, the matrix is nonsingular, forcing $c_1=\cdots=c_h=0$, which contradicts $c_j\neq0$ for each $j\in[h]$. We conclude that $h\ge d$. Because $h=d_{\rm H}(\rho,\sigma)$, every fiber of $J_{N,d}$ has minimum Hamming distance at least $d$.

Finally, $J_{N,d}$ takes values in $\mathbb F_p^{d-2}$ and thus admits at most $p^{d-2}$ distinct values. Computing its $d-2$ components directly requires $O(Nd)$ field operations.
\end{proof}

For a permutation $\pi=(\pi_1,\ldots,\pi_n)\in\mathcal S_n$, we append the sentinel symbol $N=n+1$ to form $\bar\pi \triangleq (\pi_1,\ldots,\pi_n,N)\in\mathcal S_N$. Its cyclic-successor permutation $f(\bar\pi)\in\mathcal S_N$ is defined by
\begin{equation}
f(\bar\pi)_{\bar\pi_i} \triangleq \bar\pi_{i+1},
 \qquad i\in[N],
 \qquad \bar\pi_{N+1} \triangleq \bar\pi_1.
 \label{eq:sid-successor-map}
\end{equation}
Under this definition, $f(\bar\pi)_a$ is the symbol immediately succeeding $a$ in the cyclic ordering determined by $\bar\pi$. It was proved in~\cite[Claim~1 and Theorem~2]{wang2025permutation} that $f$ is injective on permutations ending in $N$, and satisfies
\begin{equation}
 d_{\mathrm H}\bigl(f(\bar\pi),f(\bar\sigma)\bigr)
 \leq 3d_{\mathrm U}(\pi,\sigma)
 \label{eq:sid-successor-distance}
\end{equation}
for all $\pi,\sigma\in\mathcal S_n$. For $n\geq3t$, we therefore define the inner SID syndrome map
\begin{equation}
 \kappa^{\mathrm{SID}}_{n,t}(\pi)
 \triangleq J_{n+1,3t+1}\bigl(f(\bar\pi)\bigr)
 \in\mathbb F_p^{3t-1}.
 \label{eq:inner-sid-syndrome}
\end{equation}

\begin{theorem}[Polynomial-Time Inner SID Syndrome]
\label{thm:inner-sid-syndrome}
Every syndrome class of $\kappa^{\mathrm{SID}}_{n,t}$ corrects up to $t$ SIDs. The size of the ambient syndrome space satisfies
\begin{equation}
 M^{\mathrm{SID}}_{n,t} \triangleq p^{3t-1}
 <2^{3t-1}(n+1)^{3t-1}.
 \label{eq:inner-sid-range}
\end{equation}
The syndrome evaluation and the corresponding decoding algorithm can be computed in $n^{O(t)}$ time.
\end{theorem}

\begin{proof}
Suppose that distinct permutations $\pi,\sigma\in\mathcal S_n$ satisfy $\kappa^{\mathrm{SID}}_{n,t}(\pi)=\kappa^{\mathrm{SID}}_{n,t}(\sigma)$. By the injectivity of $f$ and Lemma~\ref{lem:canonical-hamming-partition}, their cyclic successors satisfy $d_{\mathrm H}\bigl(f(\bar\pi),f(\bar\sigma)\bigr)\geq3t+1$. If their $t$-SID descendant sets were to intersect, their Ulam distance would satisfy $d_{\mathrm U}(\pi,\sigma)\leq t$ by the metric characterization in Section~\ref{subsec:sidpid_model}. However, \eqref{eq:sid-successor-distance} would then imply a Hamming distance of at most $3t$, yielding a contradiction.

By Bertrand's postulate, the prime satisfies $p<2(n+1)$, which establishes \eqref{eq:inner-sid-range}. To decode a received sequence $x\in\mathcal I_{n,n-s}$ with $s\leq t$, we enumerate its $s!\binom ns\leq n^t$ SID parents in $\mathcal S_n$ as in Lemma~\ref{lem:parent-enumeration}, compute \eqref{eq:inner-sid-syndrome} for each candidate, and select the unique parent matching the given syndrome. All field operations involve operands of $O_t(\log n)$ bits, ensuring that the decoding procedure runs in polynomial time.
\end{proof}

\subsection{Algebraic Outer Partition}

We retain the canonical labeled-gap words $W_n(z)$, the interleaving mapping $\operatorname{Int}_z(\mu)$, the projected outer channel $\mathcal T_{\leq t}(z)$, and the outer syndrome $\Omega_{n,r,t}$ from Section~\ref{subsec:pid-outer}. Recall that, for $r\geq2t$, Theorem~\ref{thm:pid-outer-partition} guarantees the existence of a syndrome class $\mathcal A_{n,r,t}\subseteq\{0,\ldots,n\}^r$ with cardinality
\begin{equation}
 |\mathcal A_{n,r,t}|
 \geq\frac{(n+1)^{r-t}}{B_{r,t}Q_{r,t}},
 \label{eq:sid-reused-outer-size}
\end{equation}
while admitting uniform polynomial-time construction, ranking, and unranking algorithms. The next lemma demonstrates that this class also corrects the outer channel under SIDs.

Let $w$ be a sequence obtained from $W_n(z)$ through the deletion of $d\leq t$ symbols, each being either a copy of $D$ or a marker. Let $E\subseteq[r]$ denote the set of deleted marker labels, let $e \triangleq |E|$, and let $\delta \triangleq d-e$ be the number of deleted copies of $D$. We write the surviving markers as $[r]\setminus E=\{s_1<s_2<\cdots<s_{r-e}\}$.

\begin{lemma}[SID-to-Projected-PID Reduction]
\label{lem:sid-to-outer-projection}
Let $\Lambda_{\delta,E}(w)$ be the word obtained from $w$ by replacing the surviving markers $s_1,\ldots,s_\delta$ with copies of $D$ and relabeling $s_{\delta+k}\mapsto k$ for each $k\in[r-d]$. Then
\begin{equation*}
 \Lambda_{\delta,E}(w)\in\mathcal T_{\leq t}(z).
\end{equation*}
Consequently, if $z,z'$ belong to the same syndrome class of $\Omega_{n,r,t}$ and $W_n(z),W_n(z')$ share an SID descendant after at most $t$ deletions, then $z=z'$.
\end{lemma}

\begin{proof}
Recall that the operation defining $\mathcal T_{\leq t}(z)$ in Section~\ref{subsec:pid-outer} deletes a set $E$ of markers along with $\delta$ placeholders $D$ from $W_n(z)$, converts the first $\delta$ surviving markers into $D$, and relabels the remaining markers in increasing order. Applying these last two steps to the stable sequence $w$ yields precisely $\Lambda_{\delta,E}(w)$, proving the first assertion.

For the second assertion, suppose $w$ is a common stable descendant of $W_n(z)$ and $W_n(z')$. The set $E$, and consequently the integers $e$ and $\delta$, are determined uniquely by $w$. Thus, the word $\Lambda_{\delta,E}(w)$ belongs to both $\mathcal T_{\leq t}(z)$ and $\mathcal T_{\leq t}(z')$. Theorem~\ref{thm:pid-outer-partition} establishes that these projected output sets are mutually disjoint within each syndrome class of $\Omega_{n,r,t}$, which enforces $z=z'$.
\end{proof}

\subsection{The Full-Systematic SID Construction}

Set $r \triangleq 4t$, $B_t^{\mathrm{SID}} \triangleq B_{4t,t}$, and $Q_t^{\mathrm{SID}} \triangleq Q_{4t,t}$. Let $\mathcal A^{\mathrm{SID}}_{n,t} \triangleq \mathcal A_{n,4t,t}$ be the lexicographically first largest syndrome class provided by Theorem~\ref{thm:pid-outer-partition}. Applying \eqref{eq:sid-reused-outer-size} yields
\begin{equation}
 |\mathcal A^{\mathrm{SID}}_{n,t}|
 \geq\frac{(n+1)^{3t}}{B_t^{\mathrm{SID}}Q_t^{\mathrm{SID}}}.
 \label{eq:selected-sid-outer-size}
\end{equation}
The inner syndrome in Theorem~\ref{thm:inner-sid-syndrome} takes values in an ambient space of size $M^{\mathrm{SID}}_{n,t}=p^{3t-1}$. Comparing \eqref{eq:inner-sid-range} with \eqref{eq:selected-sid-outer-size} demonstrates that $|\mathcal A^{\mathrm{SID}}_{n,t}|\geq M^{\mathrm{SID}}_{n,t}$ holds whenever
\begin{equation}
 n\geq n_{\mathrm{SID}}(t) \triangleq
 \max\left\{3t,
 \left\lceil 2^{3t-1}B_t^{\mathrm{SID}}Q_t^{\mathrm{SID}}\right\rceil-1
 \right\}.
 \label{eq:sid-n0}
\end{equation}

For $a=(a_1,\ldots,a_{3t-1})\in\mathbb F_p^{3t-1}$, let $\widetilde a_k\in\{0,\ldots,p-1\}$ denote the standard integer representative of $a_k$. We map $a$ to the integer index
\begin{equation}
 \operatorname{idx}_p(a)
 \triangleq \sum_{k=1}^{3t-1}\widetilde a_kp^{k-1}
 \in\{0,\ldots,M^{\mathrm{SID}}_{n,t}-1\}.
 \label{eq:sid-index}
\end{equation}
Let $z_{\mathrm{SID}}^{(q)}$ denote the vector of rank $q$ in the lexicographically ordered class $\mathcal A^{\mathrm{SID}}_{n,t}$.

\begin{construction}[Systematic SID Encoder]
\label{con:systematic-sid-encoder}
Given a message permutation $\mu\in\mathcal S_n$, encoding is performed in three steps.
\begin{enumerate}[label=\arabic*)]
 \item Compute the inner syndrome $a = \kappa^{\mathrm{SID}}_{n,t}(\mu)$ and its index $q = \operatorname{idx}_p(a)$.
 \item Retrieve the vector $z = z^{(q)}_{\mathrm{SID}}\in \mathcal A^{\mathrm{SID}}_{n,t}$ by unranking $q$.
 \item Output the codeword $\Enc^{\mathrm{SID}}_{n,t}(\mu) \triangleq \Int_z(\mu)\in \mathcal S_{n+4t}$.
\end{enumerate}
\end{construction}

To decode, we utilize the outer dictionary established in Construction~\ref{con:full-systematic-pid-decoder}, where each sequence in $\bigcup_{z\in\mathcal A^{\mathrm{SID}}_{n,t}} \mathcal T_{\leq t}(z)$ is associated with its unique parent vector $z$. Uniqueness is guaranteed by Theorem~\ref{thm:pid-outer-partition}.

\begin{construction}[Systematic SID Decoder]
\label{con:systematic-sid-decoder}
Given a received sequence $y\in\mathcal I_{n+4t,n+4t-d}$ produced by at most $d\leq t$ SIDs, decoding proceeds in three steps.
\begin{enumerate}[label=\arabic*)]
 \item Identify the deleted markers $E \triangleq \{j\in[4t] : n+j\text{ does not occur in }y\}$, set $e \triangleq |E|$, and calculate $\delta \triangleq d-e$. Replace each surviving data symbol in $[n]$ with $D$ and each surviving marker $n+j$ with $j$, obtaining a stable outer sequence $w$. Construct $\widehat w \triangleq \Lambda_{\delta,E}(w)$ and consult the outer dictionary to retrieve the unique vector $z\in\mathcal A^{\mathrm{SID}}_{n,t}$ satisfying $\widehat w\in\mathcal T_{\leq t}(z)$.
 \item Determine the rank $q$ of $z$ within $\mathcal A^{\mathrm{SID}}_{n,t}$. If $q\geq M^{\mathrm{SID}}_{n,t}$, declare a decoding failure. Otherwise, invert \eqref{eq:sid-index} to recover the syndrome $a\in\mathbb F_p^{3t-1}$.
 \item Extract the data subsequence $x \triangleq y|_{[n]}\in\mathcal I_{n,n-\delta}$. Enumerate its $\delta!\binom n\delta$ SID parents in $\mathcal S_n$, and output the unique candidate $\mu'$ satisfying $\kappa^{\mathrm{SID}}_{n,t}(\mu')=a$.
\end{enumerate}
\end{construction}

\begin{theorem}[Systematic $t$-SID Correcting Code]
\label{thm:systematic-sid}
For every fixed integer $t\geq1$ and all $n\geq n_{\mathrm{SID}}(t)$, the encoder $\Enc^{\mathrm{SID}}_{n,t}$ from $\mathcal S_n$ to $\mathcal S_{n+4t}$ defined in Construction~\ref{con:systematic-sid-encoder} is full-systematic, and its codebook corrects up to $t$ symbol-invariant deletions. The code family is uniformly constructible, encodable, and decodable in $n^{O(t)}$ time. The redundancy satisfies
\begin{equation*}
\red\bigl(\Enc^{\mathrm{SID}}_{n,t}\bigr) = 4t\log n+O_t(1).
\end{equation*}
\end{theorem}

\begin{proof}
Deleting the fixed redundancy symbols $\mathcal R_{n,4t}=[n+1,n+4t]$ from $\Int_z(\mu)$ yields $\mu$, which confirms that the encoder is full-systematic.

Suppose that two codewords $\Enc^{\mathrm{SID}}_{n,t}(\mu)$ and $\Enc^{\mathrm{SID}}_{n,t}(\nu)$ produce a common output $y$ after $d\leq t$ SIDs. Let $z,z'\in\mathcal A^{\mathrm{SID}}_{n,t}$ denote the outer vectors chosen for $\mu$ and $\nu$, respectively. Processing $y$ according to the first step of Construction~\ref{con:systematic-sid-decoder} generates a common stable outer output $w$. By Lemma~\ref{lem:sid-to-outer-projection}, $z$ and $z'$ share a common projected outer sequence. Because both vectors lie in the same outer class, Theorem~\ref{thm:pid-outer-partition} enforces $z=z'$. Their ranks must therefore coincide, and the injectivity of \eqref{eq:sid-index} guarantees that
\begin{equation*}
 \kappa^{\mathrm{SID}}_{n,t}(\mu) = \kappa^{\mathrm{SID}}_{n,t}(\nu).
\end{equation*}
The data subsequence $y|_{[n]}$ is a common SID descendant of $\mu$ and $\nu$ after at most $t$ deletions. Theorem~\ref{thm:inner-sid-syndrome} then establishes that $\mu=\nu$.

This argument also confirms that each step in Construction~\ref{con:systematic-sid-decoder} produces a unique result. The inner syndrome map and its decoder run in $n^{O(t)}$ time by Theorem~\ref{thm:inner-sid-syndrome}. Because $r=4t$, Theorem~\ref{thm:pid-outer-partition} constructs the outer class and performs rank conversions in $n^{O(t)}$ time. Enumerating the projected outputs of its vectors constructs the dictionary within the same time complexity. Finally, the parent candidate list has size $\delta!\binom n\delta\leq n^t$. Therefore, construction, encoding, and decoding all operate in $n^{O(t)}$ time.

Lastly, the code redundancy evaluates to
\begin{equation*}
\red\bigl(\Enc^{\mathrm{SID}}_{n,t}\bigr)
 = \log\frac{(n+4t)!}{n!}
 = \sum_{j=1}^{4t}\log(n+j)
 = 4t\log n+O_t(1).
\end{equation*}\qedhere
\end{proof}

%% file: main_paper/multipermutation.tex
\section{Extension to Multipermutations}
\label{sec:multipermutation}

In this section, we extend the systematic coding constructions from Sections~\ref{sec:pid} and~\ref{sec:sid} to fixed-composition multipermutations. The presence of repeated symbols introduces two key technical challenges. First, PID standardization must preserve ties among identical symbols rather than resolving them arbitrarily. Second, under the PID model, symbols that are completely deleted lead to an alphabet reduction that conceals which symbol classes have vanished. We address this uncertainty by introducing a profile-dependent ambiguity parameter. This formulation naturally interpolates between the permutation regime and the high-multiplicity multipermutation regime. In particular, for a $\lambda$-regular profile with $\lambda>t$, no symbol class can be completely deleted under at most $t$ deletions, causing the PID and SID channels to coincide.

\subsection{Deletion Models in Multipermutation and the Stable Lift}

Let $\boldsymbol m=(m_1,\ldots,m_k)\in\mathbb Z_{>0}^{k}$, let $n \triangleq \sum_{a=1}^{k}m_a$, and define the fixed-composition multipermutation space
\begin{equation*}
 \cM(\boldsymbol m) \triangleq \bigl\{x=(x_1,\ldots,x_n)\in[k]^n : N_a(x)=m_a \text{ for every } a\in[k]\bigr\},
\end{equation*}
where $N_a(x) \triangleq |\{i\in[n] : x_i=a\}|$. Its cardinality is given by
\begin{equation}
       |\cM(\boldsymbol m)| = \frac{n!}{\prod_{a=1}^{k}m_a!}.
 \label{mp:eq:fixed-composition-cardinality}
\end{equation}
When $k=1$, this space contains only a single word, rendering all coding statements trivial. We therefore assume $k\ge2$ in what follows.

For a nonempty word $u=(u_1,\ldots,u_s)$ over a linearly ordered alphabet, its tie-preserving standardization is defined by
\begin{equation*}
  \mstd(u)_i \triangleq 1+\bigl|\{a\in\supp(u) : a<u_i\}\bigr|,
  \qquad i\in[s].
\end{equation*}
Under this mapping, identical symbols remain identical, while the distinct symbols appearing in $u$ are relabeled in strictly increasing order by $1,\ldots,|\supp(u)|$. For an index set $I\subseteq[n]$, let $x_{[n]\setminus I}$ denote the surviving subsequence, with positions retained in increasing order. We formally define the deletion operations on multipermutations as
\begin{equation*}
 \operatorname{del}^{\SID}_{I}(x) \triangleq x_{[n]\setminus I},\qquad
 \operatorname{del}^{\PID}_{I}(x) \triangleq \mstd\bigl(x_{[n]\setminus I}\bigr).
\end{equation*}
For $\mathsf x\in\{\SID,\PID\}$ and $0\le s\le n$, we define the descendant sets
\begin{equation*}
 \cD_s^{\mathsf x}(x) \triangleq \bigl\{\operatorname{del}^{\mathsf x}_{I}(x) : I\subseteq[n],\, |I|=s\bigr\},\qquad
 \cD_{\le t}^{\mathsf x}(x) \triangleq \bigcup_{s=0}^{t}\cD_s^{\mathsf x}(x).
\end{equation*}
A subset of $\cM(\boldsymbol m)$ corrects up to $t$ deletions of type $\mathsf x$ if the corresponding sets $\cD_{\le t}^{\mathsf x}(x)$ are pairwise disjoint. The following transitivity property provides the multipermutation counterpart to Lemma~\ref{lem:std-transitive}.

\begin{lemma}[Transitivity of Tie-Preserving Standardization]
\label{mp:lem:mstd-transitivity}
Let $u=(u_1,\ldots,u_s)$ be a word and let $J=\{j_1<\cdots<j_h\}\subseteq[s]$. Then
\begin{equation*}
 \mstd\bigl(\mstd(u)_{j_1},\ldots,\mstd(u)_{j_h}\bigr)
 =\mstd(u_{j_1},\ldots,u_{j_h}).
\end{equation*}
Consequently, for either deletion model, two words of length $n$ share a common descendant after at most $t$ deletions if and only if they share a common descendant after exactly $t$ deletions, provided $n>t$.
\end{lemma}

\begin{proof}
For any indices $i,j\in J$, the comparisons $u_i<u_j$ and $u_i=u_j$ hold if and only if the corresponding relations hold between $\mstd(u)_i$ and $\mstd(u)_j$. Hence, the two selected subsequences induce the same total preorder, and their tie-preserving standardizations are identical. The exact-$t$ reduction is immediate for SIDs. For PIDs, given a common $s$-PID descendant with $s<t$, deleting the same additional $t-s$ positions from both sequences and applying the displayed identity establishes the claim.
\end{proof}

To connect multipermutation PIDs with the permutation PID framework studied in Section~\ref{sec:pid}, we introduce a canonical stable lift. For a word $u=(u_1,\ldots,u_s)$, we define
\begin{equation}
\Lift(u)_i \triangleq \sum_{a<u_i}N_a(u)+\bigl|\{j\le i : u_j=u_i\}\bigr|,
 \qquad i\in[s].
 \label{mp:eq:stable-lift}
\end{equation}
The occurrences of each symbol are therefore labeled in strictly increasing order from left to right. For the composition profile $\boldsymbol m$, set $M_0 \triangleq 0$ and $M_a \triangleq \sum_{b=1}^{a}m_b$, and define the consecutive blocks $B_a \triangleq \{M_{a-1}+1,\ldots,M_a\}$ for each $a\in[k]$. We then define the subspace of permutations
\begin{equation*}
 \cL_{\boldsymbol m} \triangleq \bigl\{\pi\in\mathcal S_n : \pi|_{B_a}\text{ is increasing for every }a\in[k]\bigr\},
\end{equation*}
where $\pi|_{B_a}$ denotes the subsequence of $\pi$ restricted to the symbols in $B_a$.

\begin{example}
Let $k=3$, $\boldsymbol m=(2,2,2)$, and $u=(2,1,3,1,2,3)\in\mathcal M(\boldsymbol m)$. Reading from left to right within each symbol class, the occurrences of symbols $1,2,3$ receive the label blocks $\{1,2\}$, $\{3,4\}$, and $\{5,6\}$, respectively. Hence, $\Lift(u)=(3,1,5,2,4,6)\in\mathcal L_{\boldsymbol m}$.

Now retain the positions $J=\{2,3,4,6\}$, which deletes both occurrences of symbol $2$. Then $u_J=(1,3,1,3)$ and $\mstd(u_J)=(1,2,1,2)$, yielding $\Lift\bigl(\mstd(u_J)\bigr)=(1,3,2,4)$.

On the other hand, $\Lift(u)_J=(1,5,2,6)$ and $\std\bigl(\Lift(u)_J\bigr)=(1,3,2,4)$. Thus, $\Lift\bigl(\mstd(u_J)\bigr)=\std\bigl(\Lift(u)_J\bigr)$, illustrating the commutation relation formalized below.
\end{example}

\begin{lemma}[Stable-Lift Bijection and Commutation]
\label{mp:lem:stable-lift}
The map $\Lift$ restricts to a bijection from $\cM(\boldsymbol m)$ to $\cL_{\boldsymbol m}$. Moreover, for every word $u$ and every subset of surviving positions $J\subseteq[|u|]$,
\begin{equation}
\Lift\bigl(\mstd(u_J)\bigr)=\operatorname{std}\bigl(\Lift(u)_J\bigr).\label{eq:lift-commutation}
\end{equation}
In particular, if $y\in\cD_s^{\PID}(x)$, then $\Lift(y)\in\cD_s^{\PID}(\Lift(x))$ under the standard permutation PID definition.
\end{lemma}

\begin{proof}
For any $x\in\cM(\boldsymbol m)$, the occurrences of symbol $a$ receive precisely the labels in $B_a$ in their order of appearance, ensuring $\Lift(x)\in\cL_{\boldsymbol m}$. Conversely, replacing every symbol of $B_a$ in a permutation $\pi\in\cL_{\boldsymbol m}$ by $a$ recovers a unique word in $\cM(\boldsymbol m)$, and the block-increasing condition guarantees that lifting this word yields $\pi$. This proves bijectivity.

For any two surviving positions $i,j\in J$, the relation $\Lift(u)_i<\Lift(u)_j$ holds if and only if either $u_i<u_j$, or $u_i=u_j$ with $i<j$. This coincides with the ordering employed by the stable lift of $\mstd(u_J)$, where distinct surviving symbol classes are ordered by value, and multiple occurrences within the same class are ordered by coordinate. Thus, the two permutations in \eqref{eq:lift-commutation} induce the same strict total order, proving their equality.
\end{proof}

We next quantify the structural ambiguity introduced by repeated symbols under PID standardization. For a subset $A\subseteq[k]$, let $w_{\boldsymbol m}(A) \triangleq \sum_{a\in A}m_a$, and for an integer $s\ge0$, we define
\begin{align*}
 \mathcal V_s(\boldsymbol m) &\triangleq \{A\subseteq[k] : w_{\boldsymbol m}(A)\le s\},\\
 B_s(\boldsymbol m) &\triangleq |\mathcal V_s(\boldsymbol m)|,\\
 \eta_s(\boldsymbol m) &\triangleq \max\{|A| : A\in\mathcal V_s(\boldsymbol m)\}.
\end{align*}
Because the empty set is included, we have $B_s(\boldsymbol m)\ge1$ and $\eta_s(\boldsymbol m)\ge0$. We also set
\begin{equation*}
 \varepsilon(\boldsymbol m) \triangleq \mathbf 1\left\{\min_{a\in[k]}m_a=1\right\}.
\end{equation*}

\begin{lemma}[Profile Ambiguity and Parent Enumeration]
\label{mp:lem:parent-enumeration}
For any fixed integer $s$,
\begin{equation}
 B_s(\boldsymbol m)\le\sum_{j=0}^{\eta_s(\boldsymbol m)}\binom kj\le (s+1)n^{\eta_s(\boldsymbol m)}.
 \label{eq:B-upper}
\end{equation}
Furthermore, the following two properties hold.
\begin{enumerate}[label=\textup{(\roman*)}]
\item Every sequence of length $n-s$ produced by SIDs has at most $s!\binom ns$ parents in $\cM(\boldsymbol m)$.
\item Every standardized sequence of length $n-s$ produced by PIDs has at most $B_s(\boldsymbol m)\,s!\binom ns$ parents in $\cM(\boldsymbol m)$.
\end{enumerate}
All parents in either case can be enumerated in $n^{O(s)}$ time.
\end{lemma}

\begin{proof}
Every member of $\mathcal V_s(\boldsymbol m)$ has cardinality at most $\eta_s(\boldsymbol m)\le s$, which establishes the first inequality in \eqref{eq:B-upper}. Because $k\le n$ and the sum contains at most $s+1$ terms, the second inequality follows.

For an SID descendant $y$, the missing multiplicities $\delta_a \triangleq m_a-N_a(y)$ are uniquely known and sum to $s$. Every parent is formed by choosing $s$ insertion positions and placing the missing multiset in some order. The number of such ordered descriptions is at most $\binom ns s!$.

Now let $y$ be a standardized PID descendant, and let $A$ denote the subset of original symbol classes that were completely eliminated. Necessarily, $A\in\mathcal V_s(\boldsymbol m)$ and $|A|=k-|\supp(y)|$. For each candidate set $A$, there exists a unique order-preserving bijection $\phi_A$ from $[|\supp(y)|]$ to $[k]\setminus A$. Replacing each symbol $j$ in $y$ by $\phi_A(j)$ yields an uncompressed word $y^{(A)}$. We retain $A$ only when the deficits $\delta_a \triangleq m_a-N_a(y^{(A)})$ are all non-negative and sum to $s$. The parents corresponding to this set $A$ are then generated by inserting the deficit multiset into $y^{(A)}$, which involves at most $\binom ns s!$ possibilities. Enumerating subsets of size at most $s$, insertion positions, and insertion orderings requires $n^{O(s)}$ time.
\end{proof}

\subsection{From a Single-Deletion Code to a $t$-Deletion Syndrome}
\label{mp:subsec:amplification}

In this subsection, we use the local algebraic fingerprint of Lemma~\ref{lem:local-fingerprint} to extend single-deletion-correcting codes to multiple deletions. The following lemma formalizes this two-stage amplification.

\begin{lemma}[Residual Graph Amplification]
\label{mp:lem:residual-amplification}
Let $\mathsf x\in\{\SID,\PID\}$ and let $b$ be a polynomial-time syndrome mapping from $\cM(\boldsymbol m)$ to $\mathcal B$ such that each syndrome class of $b$ corrects a single $\mathsf x$-deletion. Suppose that $|\mathcal B|\le M_0$, and that for every exact-$t$ output $y$ and every $\beta\in\mathcal B$,
\begin{equation}
 \bigl|\{x\in\cM(\boldsymbol m) : y\in\cD_t^{\mathsf x}(x),\, b(x)=\beta\}\bigr| \le L.
 \label{mp:eq:residual-parent-assumption}
\end{equation}
Assume also that the parents in \eqref{mp:eq:residual-parent-assumption} can be enumerated in $n^{O(t)}$ time. Setting $K \triangleq \max\left\{n,\binom nt L\right\}$, there exists a polynomial-time syndrome $\mathsf K_{b,t}$ from $\cM(\boldsymbol m)$ to $\mathcal S$ whose syndrome classes correct up to $t$ $\mathsf x$-deletions, with ambient range size satisfying
\begin{equation}
|\mathcal S|<64M_0K^2.
 \label{mp:eq:generic-range}
\end{equation}
The corresponding decoder runs in $n^{O(t)}$ time.
\end{lemma}

\begin{proof}
For each $\beta\in\mathcal B$, we construct the residual graph on the vertex set $b^{-1}(\beta)$, where two distinct words are adjacent if and only if they share an exact-$t$ $\mathsf x$-descendant. Each vertex generates at most $\binom nt$ exact-$t$ outputs, and each output admits at most $L$ compatible parents within the same $b$-syndrome class. Hence, every residual graph has maximum degree at most $K$, and its neighborhoods are enumerable in $n^{O(t)}$ time.

If $|\cM(\boldsymbol m)|=1$, the result is trivial. Otherwise, we order $\cM(\boldsymbol m)$ lexicographically and let $c_0(x)$ denote the zero-based lexicographic rank of $x$. This provides an injective, and therefore proper, coloring with $2\le U=|\cM(\boldsymbol m)|\le n!$, which can be computed in polynomial time. Let $p_0$ and $p_1$ be the smallest primes satisfying $p_0>2nK$ and $p_1>4K$. By Bertrand's postulate, $p_0<4nK$ and $p_1<8K$. Since $p_0>n$, we have $\left\lceil\log_{p_0}U\right\rceil\le n$ and $K(n-1)<p_0$. Applying Lemma~\ref{lem:local-fingerprint} to $c_0$ over $\mathbb F_{p_0}$ produces a proper coloring $c_1$ with at most $p_0^2$ colors. Because $K\ge n$, we have $p_0^2<16n^2K^2\le16K^4<p_1^4$, which implies $\lceil\log_{p_1}(p_0^2)\rceil\le4$. A second application of Lemma~\ref{lem:local-fingerprint} over $\mathbb F_{p_1}$ is valid because $3K<p_1$, yielding a proper coloring $c_2$ with fewer than $p_1^2<64K^2$ colors.

We define $\mathsf K_{b,t}(x) \triangleq (b(x),c_2(x))$. Within each $b$-class, equality under $c_2$ excludes graph adjacency, ensuring that codewords with identical syndromes share no exact-$t$ descendant. By Lemma~\ref{mp:lem:mstd-transitivity}, this guarantees correction of up to $t$ deletions. The range bound is given by \eqref{mp:eq:generic-range}. Computing both fingerprints requires only radius-two residual neighborhoods, which can be enumerated in $n^{O(t)}$ time. To decode, we enumerate all parents of the received word and retain the unique candidate matching the syndrome.
\end{proof}

\subsection{Inner Syndromes for Multipermutation SIDs}
\label{mp:subsec:sid-inner}

For a sequence $u=(u_1,\ldots,u_n)\in[k]^n$, its non-decreasing signature $\alpha(u)\in\{0,1\}^{n-1}$ is defined by
\begin{equation*}
      \alpha(u)_i \triangleq \mathbf 1\{u_i\le u_{i+1}\},
      \qquad i\in[n-1],
\end{equation*}
and we define the associated Varshamov-Tenengolts syndrome as
\begin{equation*}
    \VT_n(u) \triangleq \sum_{i=1}^{n-1}i\alpha(u)_i\pmod n.
\end{equation*}
Tenengolts' non-binary construction established that constraining both $\VT_n(u)$ and $\sum_i u_i\pmod k$ yields a single-deletion-correcting code~\cite{tenengolts1984nonbinary}. On the multipermutation space $\cM(\boldsymbol m)$, the symbol sum is fixed by the composition profile. Consequently, each syndrome class of $\VT_n$ inside $\cM(\boldsymbol m)$ corrects a single SID, and $\VT_n$ takes at most $n$ distinct values.

\begin{lemma}[Residual SID Parent List]
\label{mp:lem:sid-residual-list}
For every sequence $y\in[k]^{n-t}$ and every $a\in\mathbb Z_n$,
\begin{equation*}
 \bigl|\{x\in\cM(\boldsymbol m) : y\in\cD_t^{\SID}(x),\, \VT_n(x)=a\}\bigr|
 \le t!\binom{n-1}{t-1}.
\end{equation*}
Furthermore, this set can be enumerated in $n^{O(t)}$ time.
\end{lemma}

\begin{proof}
Let $\Delta$ denote the missing multiset of size $t$ determined by $y$ and $\boldsymbol m$. For any parent $x$, select one of the $t$ deleted elements as the first deletion, and let $\rho$ denote the resulting word of length $n-1$. Then $y$ is obtained from $\rho$ via $t-1$ SIDs. To enumerate all valid intermediate sequences $\rho$, we select the symbol of the first deleted element and insert the remaining $t-1$ elements of $\Delta$ into $y$. This provides at most $t\binom{n-1}{t-1}(t-1)!=t!\binom{n-1}{t-1}$ configurations. For each fixed $\rho$, the single-deletion property of $\VT_n$ permits at most one parent $x$ with $\VT_n(x)=a$. This enumeration directly gives an $n^{O(t)}$-time algorithm.
\end{proof}

\begin{theorem}[Inner Multipermutation SID Syndrome]
\label{mp:thm:sid-inner}
For every fixed integer $t\ge1$, every composition profile $\boldsymbol m$, and all $n>t$, there exists an inner syndrome map $\mathsf K^{\SID}_{\boldsymbol m,t}$ from $\cM(\boldsymbol m)$ to $\mathcal S^{\SID}_{\boldsymbol m,t}$ whose syndrome classes correct up to $t$ SIDs, satisfying
\begin{equation*}
       |\mathcal S^{\SID}_{\boldsymbol m,t}| = O_t(n^{4t-1}).
\end{equation*}
The syndrome evaluation and its corresponding decoder run in $n^{O(t)}$ time.
\end{theorem}

\begin{proof}
Apply Lemma~\ref{mp:lem:residual-amplification} with $b=\VT_n$, $M_0=n$, and $L=t!\binom{n-1}{t-1}=O_t(n^{t-1})$. Then $K=O_t(n^{2t-1})$, and the bound in \eqref{mp:eq:generic-range} yields $64M_0K^2=O_t(n^{4t-1})$.
\end{proof}

\subsection{Profile-Sensitive Inner Syndromes for Multipermutation PIDs}
\label{mp:subsec:pid-inner}

When every multiplicity in $\boldsymbol m$ is at least two, a single deletion cannot eliminate an entire symbol class. Hence, tie-preserving standardization leaves all symbol values unchanged, rendering a single PID identical to a single SID. When singleton classes exist, we instead apply the stable lift and utilize the two-coordinate single-PID syndrome $\nu_n$ from Section~\ref{subsec:inner-pid-syndrome}. We define the base syndrome
\begin{equation}
 b^{\PID}_{\boldsymbol m}(x) \triangleq
 \begin{cases}
   \VT_n(x),&\varepsilon(\boldsymbol m)=0,\\
   \nu_n(\Lift(x)),&\varepsilon(\boldsymbol m)=1.
 \end{cases}
 \label{mp:eq:pid-base-syndrome}
\end{equation}
The range of $b^{\PID}_{\boldsymbol m}$ contains at most $n^{1+\varepsilon(\boldsymbol m)}$ values.

\begin{lemma}[One-PID Base Partition]
\label{mp:lem:one-pid-base}
Every syndrome class of $b^{\PID}_{\boldsymbol m}$ corrects a single multipermutation PID.
\end{lemma}

\begin{proof}
If $\varepsilon(\boldsymbol m)=0$, all symbol classes survive one deletion, so that $\operatorname{del}^{\PID}_I(x)=\operatorname{del}^{\SID}_I(x)$ for $|I|=1$. The claim then follows from Tenengolts' construction. If $\varepsilon(\boldsymbol m)=1$ and two words share a single-PID output, Lemma~\ref{mp:lem:stable-lift} guarantees that their stable lifts share a single-PID output. The single-PID capability of $\nu_n$ separates the two lifted permutations, and the bijectivity of $\Lift$ separates the original words.
\end{proof}

\begin{lemma}[Residual PID Parent List]
\label{mp:lem:pid-residual-list}
For every standardized sequence $y$ of length $n-t$ and every syndrome value $\beta$ of $b^{\PID}_{\boldsymbol m}$,
\begin{equation*}
 \bigl|\{x\in\cM(\boldsymbol m) : y\in\cD_t^{\PID}(x),\, b^{\PID}_{\boldsymbol m}(x)=\beta\}\bigr|
 \le B_t(\boldsymbol m)\,t!\binom{n-1}{t-1}.
\end{equation*}
Furthermore, this set can be enumerated in $n^{O(t)}$ time.
\end{lemma}

\begin{proof}
Fix a parent $x$ and let $A$ denote the set of its symbol classes that were completely deleted in forming $y$. Then $A\in\mathcal V_t(\boldsymbol m)$ and $|A|=k-|\supp(y)|$. As in the proof of Lemma~\ref{mp:lem:parent-enumeration}, we uncompress $y$ via the order-preserving bijection from its normalized alphabet to $[k]\setminus A$, denoting the result by $y^{(A)}$. Relative to the original profile, $y^{(A)}$ is missing a multiset $\Delta_A$ of exactly $t$ elements.

Select one of these $t$ elements as the first deletion from $x$. Prior to standardization, the corresponding sequence of length $n-1$ is obtained by inserting the remaining $t-1$ elements of $\Delta_A$ into $y^{(A)}$. Applying $\mstd$ produces an intermediate standardized word $\rho$ satisfying $\rho\in\cD_1^{\PID}(x)$ and $y\in\cD_{t-1}^{\PID}(\rho)$. For a fixed set $A$, there are at most $t\binom{n-1}{t-1}(t-1)!=t!\binom{n-1}{t-1}$ configurations for $\rho$, and there are at most $B_t(\boldsymbol m)$ choices for $A$. Finally, for each fixed $\rho$, Lemma~\ref{mp:lem:one-pid-base} ensures that at most one parent $x$ possesses the base syndrome $\beta$. Enumerating these configurations provides the claimed algorithm.
\end{proof}

\begin{theorem}[Profile-Sensitive Inner PID Syndrome]
\label{mp:thm:pid-inner-direct}
For every fixed integer $t\ge1$, there exists a polynomial-time syndrome $\mathsf K^{\PID,\mathrm{dir}}_{\boldsymbol m,t}$ from $\cM(\boldsymbol m)$ to $\mathcal S^{\PID,\mathrm{dir}}_{\boldsymbol m,t}$ whose syndrome classes correct up to $t$ multipermutation PIDs, with ambient range size satisfying
\begin{equation}
 |\mathcal S^{\PID,\mathrm{dir}}_{\boldsymbol m,t}|=O_t\!\left(B_t(\boldsymbol m)^2n^{4t-1+\varepsilon(\boldsymbol m)}\right)=O_t\!\left(n^{4t-1+\varepsilon(\boldsymbol m)+2\eta_t(\boldsymbol m)}\right).
 \label{mp:eq:direct-pid-range-exponent}
\end{equation}
The corresponding decoder runs in $n^{O(t)}$ time.
\end{theorem}

\begin{proof}
Apply Lemma~\ref{mp:lem:residual-amplification} to $b=b^{\PID}_{\boldsymbol m}$. Lemma~\ref{mp:lem:pid-residual-list} gives $L=O_t\bigl(B_t(\boldsymbol m)n^{t-1}\bigr)$ and $K=O_t\bigl(B_t(\boldsymbol m)n^{2t-1}\bigr)$, while the base range has cardinality at most $n^{1+\varepsilon(\boldsymbol m)}$. Substituting these parameters into \eqref{mp:eq:generic-range} and applying Lemma~\ref{mp:lem:parent-enumeration} establishes \eqref{mp:eq:direct-pid-range-exponent}.
\end{proof}

A profile-independent alternative can also be obtained directly from the permutation PID syndrome. By Lemma~\ref{mp:lem:stable-lift}, every common multipermutation PID descendant lifts to a common ordinary permutation PID descendant. Hence, the lifted syndrome defined by $\mathsf K^{\PID,\mathrm{lift}}_{\boldsymbol m,t}(x) \triangleq \kappa_{n,t}(\Lift(x))$, where $\kappa_{n,t}$ is the syndrome from Theorem~\ref{thm:inner-pid-syndrome}, also corrects $t$ multipermutation PIDs with $O_t(n^{6t-2})$ possible values. In our final construction, we choose whichever syndrome minimizes the required number of redundancy markers.

\subsection{General Fixed-Composition Full-Systematic Construction}
\label{mp:subsec:systematic}

For an integer $r\ge2t$, let $\mathcal A_{n,r,t}\subseteq\{0,\ldots,n\}^r$ denote the canonical outer code class from Theorem~\ref{thm:pid-outer-partition}. We utilize the established properties that $|\mathcal A_{n,r,t}|=\Omega_{r,t}(n^{r-t})$, that its projected output sets $\cT_{\le t}(z)$ are pairwise disjoint, and that it supports polynomial-time construction, ranking, and unranking.

We augment the message profile by appending $r$ singleton marker classes,
\begin{equation*}
   \boldsymbol m^{+r} \triangleq (m_1,\ldots,m_k,
\underbrace{1,\ldots,1}_{r\text{ times}}).
\end{equation*}
Recall that the labeled-gap word $W_n(z)$ from Section~\ref{subsec:pid-outer} contains $n$ copies of the placeholder $D$ and one copy of each marker label $1,\ldots,r$. For $x=(x_1,\ldots,x_n)\in\cM(\boldsymbol m)$, we define the interleaved sequence $\Int^{\boldsymbol m}_{z}(x) \in\cM(\boldsymbol m^{+r})$ by replacing the $i$-th occurrence of $D$ in $W_n(z)$ with $x_i$, and replacing marker $j$ with the symbol $k+j$. Deleting the fixed marker symbols $\{k+1,\ldots,k+r\}$ recovers $x$, ensuring that the encoding is full-systematic.

\begin{lemma}[Compatibility With Outer Channel]
\label{mp:lem:outer-compatibility}
Let $c=\Int^{\boldsymbol m}_{z}(x)$ and let $y$ be obtained from $c$ after $d\le t$ deletions.
\begin{enumerate}[label=\textup{(\roman*)}]
\item Under SIDs, replace each surviving data symbol in $[k]$ with $D$ and each surviving marker $k+j$ with $j$, obtaining a sequence $w$. Let $E$ denote the set of deleted marker labels, let $e \triangleq |E|$, and calculate $\delta \triangleq d-e$. Replacing the $\delta$ smallest surviving marker labels in $w$ with $D$ and relabeling the remaining markers increasingly by $1,\ldots,r-d$ yields a word in $\cT_{\le t}(z)$.
\item Under PIDs, we have $\Lift(c)=\Int_z(\Lift(x))$ and $\Theta_n(\Lift(y))\in\cT_{\le t}(z)$, where $\Theta_n$ is the projection defined in Section~\ref{subsec:pid-outer}.
\end{enumerate}
Consequently, in either channel, the outer dictionary from Theorem~\ref{thm:pid-outer-partition} recovers $z$ uniquely whenever $z\in\mathcal A_{n,r,t}$.
\end{lemma}

\begin{proof}
Part (i) follows from the SID-to-projected-PID reduction in Lemma~\ref{lem:sid-to-outer-projection}. The stable outer output reveals the deleted marker set $E$ and the number $\delta$ of deleted data symbols, and the masking and relabeling steps reproduce the definition of $\cT_{\le t}(z)$.

For part (ii), all data symbols precede all marker symbols in value. The stable lift assigns labels $1,\ldots,n$ to the data occurrences identically to $\Lift(x)$, and assigns the label $n+j$ to marker $k+j$. This establishes $\Lift(c)=\Int_z(\Lift(x))$. By Lemma~\ref{mp:lem:stable-lift}, $\Lift(y)$ is a permutation-PID descendant of $\Lift(c)$, so $\Theta_n(\Lift(y))\in\cT_{\le t}(z)$ follows from Lemma~\ref{lem:pid-projection}. Disjointness of the projected output sets inside $\mathcal A_{n,r,t}$ ensures unique identification of $z$ in both cases.
\end{proof}

\begin{lemma}[Common PID Output Induces Common Data Subsequence]
\label{mp:lem:common-data-pid}
Let $x,x'\in\cM(\boldsymbol m)$ and $z,z'\in\{0,\ldots,n\}^r$. If $\Int^{\boldsymbol m}_{z}(x)$ and $\Int^{\boldsymbol m}_{z'}(x')$ produce a common output after $d\le t$ PIDs, then $x$ and $x'$ share a common multipermutation PID descendant after at most $d$ deletions.
\end{lemma}

\begin{proof}
Let $y$ denote the common standardized output. Under the first hypothesis, let $e$ markers and $\delta \triangleq d-e$ data occurrences be deleted, and let $a$ be the number of surviving data classes. Under the second hypothesis, let $e'$, $\delta' \triangleq d-e'$, and $a'$ be defined analogously. Because the total number of distinct symbols in $y$ is identical in both cases, we have $a+r-e=a'+r-e'$. Without loss of generality, assume $a\le a'$ and define $h \triangleq a'-a$. This implies $e'=e+h$ and $\delta=\delta'+h$.

Under the first explanation, the symbols $1,\ldots,a$ in $y$ represent data classes, while all larger symbols represent surviving singleton markers. Under the second explanation, the symbols $1,\ldots,a+h$ represent data classes. Hence, the symbols $a+1,\ldots,a+h$ are markers in the first case but data symbols in the second. Each such symbol appears exactly once in $y$ because marker classes are singletons. Restricting $y$ to its entries in $[a]$ yields a sequence $v$. Under the first explanation, $v$ is the standardized data output obtained from $x$ after $\delta$ deletions. Under the second explanation, retaining data entries in $[a+h]$ yields a $\delta'$-PID output of $x'$, from which deleting the single occurrence of each class $a+1,\ldots,a+h$ reproduces $v$ by Lemma~\ref{mp:lem:mstd-transitivity}. Thus, $v$ is also a PID descendant of $x'$ after $\delta'+h=\delta\le d$ deletions, which completes the proof.
\end{proof}

We select the number of redundancy markers as
\begin{align}
 r_{\SID}(t) &\triangleq 5t,
 \label{mp:eq:r-sid}\\
 r_{\PID}(\boldsymbol m,t) &\triangleq \min\bigl\{7t-1,\,5t+\varepsilon(\boldsymbol m)+2\eta_t(\boldsymbol m)\bigr\}.
 \label{mp:eq:r-pid}
\end{align}
For the PID construction, we employ the lifted syndrome $\mathsf K^{\PID,\mathrm{lift}}_{\boldsymbol m,t}(x)=\kappa_{n,t}(\Lift(x))$ when the first term in \eqref{mp:eq:r-pid} is chosen, and the direct syndrome from Theorem~\ref{mp:thm:pid-inner-direct} otherwise. Let $M_{\mathsf x}$ denote the cardinality of the chosen ambient syndrome space, equipped with a standard mixed-radix indexing onto $\{0,\ldots,M_{\mathsf x}-1\}$. For $\mathsf x\in\{\SID,\PID\}$, we set $r \triangleq r_{\mathsf x}$ and order $\mathcal A_{n,r,t}$ lexicographically, letting $z_{\mathsf x}^{(j)}$ denote the vector of rank $j$.

\begin{construction}[Systematic Multipermutation Encoder]
\label{mp:con:encoder}
Given a message $x\in\cM(\boldsymbol m)$, encoding is performed in three steps.
\begin{enumerate}[label=\arabic*)]
\item Compute the selected inner syndrome and evaluate its mixed-radix index $j\in\{0,\ldots,M_{\mathsf x}-1\}$.
\item Retrieve the vector $z=z_{\mathsf x}^{(j)}\in\mathcal A_{n,r,t}$ by unranking the index $j$.
\item Output the codeword $\Enc^{\mathsf x}_{\boldsymbol m,t}(x) \triangleq \Int^{\boldsymbol m}_{z}(x)\in\cM(\boldsymbol m^{+r})$.
\end{enumerate}
\end{construction}

The choice of $r$ provides an additional factor of $n$ beyond the inner syndrome exponent. Specifically,
\begin{align*}
 r_{\SID}-t &= 4t > 4t-1,\\
 (7t-1)-t &= 6t-1 > 6t-2,\\
 \bigl(5t+\varepsilon+2\eta_t\bigr)-t &= 4t+\varepsilon+2\eta_t > 4t-1+\varepsilon+2\eta_t.
\end{align*}
Combined with the bound $|\mathcal A_{n,r,t}|=\Omega_{r,t}(n^{r-t})$ from Theorem~\ref{thm:pid-outer-partition}, these inequalities ensure that $|\mathcal A_{n,r,t}|\ge M_{\mathsf x}$ for all sufficiently large $n$, verifying that Construction~\ref{mp:con:encoder} is well defined.

The corresponding decoding procedure operates as follows. Under SIDs, we use Lemma~\ref{mp:lem:outer-compatibility}(i) to recover $z$, compute its rank to retrieve the inner syndrome, extract the received data subsequence, enumerate its SID parents in $\cM(\boldsymbol m)$, and output the unique parent matching the syndrome. Under PIDs, we lift the received sequence and apply Lemma~\ref{mp:lem:outer-compatibility}(ii) to recover $z$ and the inner syndrome. We then enumerate all PID parents of the received word in $\cM(\boldsymbol m^{+r})$ using Lemma~\ref{mp:lem:parent-enumeration}, retaining only those candidates of the form $\Int^{\boldsymbol m}_{z}(x')$ whose data component $x'$ matches the recovered inner syndrome.

\begin{theorem}[Systematic Multipermutation Codes]
\label{mp:thm:main}
Fix an integer $t\ge1$. For every composition profile $\boldsymbol m$ of length $n$ and all sufficiently large $n$, Construction~\ref{mp:con:encoder} provides the following codes.
\begin{enumerate}[label=\textup{(\roman*)}]
\item A full-systematic encoder $\Enc^{\SID}_{\boldsymbol m,t}$ from $\cM(\boldsymbol m)$ to $\cM(\boldsymbol m^{+5t})$ correcting up to $t$ SIDs.
\item A full-systematic encoder $\Enc^{\PID}_{\boldsymbol m,t}$ from $\cM(\boldsymbol m)$ to $\cM\bigl(\boldsymbol m^{+r_{\PID}(\boldsymbol m,t)}\bigr)$ correcting up to $t$ PIDs.
\end{enumerate}
Both code families are uniformly constructible, encodable, and decodable in $n^{O(t)}$ time. Their redundancies satisfy
\begin{equation}
\red\bigl(\Enc^{\SID}_{\boldsymbol m,t}\bigr) = 5t\log n+O_t(1), \qquad
\red\bigl(\Enc^{\PID}_{\boldsymbol m,t}\bigr) = r_{\PID}(\boldsymbol m,t)\log n+O_t(1).
\label{eq:red_mp}
\end{equation}
\end{theorem}

\begin{proof}
Suppose first that two SID codewords share an output sequence after at most $t$ deletions. By Lemma~\ref{mp:lem:outer-compatibility}(i), their outer vectors coincide, implying that their inner syndromes are identical. Restricting the common stable output to $[k]$ yields a common SID descendant of the two messages after at most $t$ deletions. Theorem~\ref{mp:thm:sid-inner} then ensures that the messages are identical.

Next, suppose two PID codewords share an output sequence. By Lemma~\ref{mp:lem:outer-compatibility}(ii), their outer vectors and selected inner syndromes are identical. Lemma~\ref{mp:lem:common-data-pid} guarantees a common multipermutation PID descendant after at most $t$ deletions. Because both inner syndrome constructions correct $t$ PIDs (by Theorem~\ref{mp:thm:pid-inner-direct} for the direct syndrome, and by Lemma~\ref{mp:lem:stable-lift} with Theorem~\ref{thm:inner-pid-syndrome} for the lifted syndrome), the two messages must be equal.

For fixed $t$, we have $r\le7t-1$. Computing inner neighborhoods, parent lists, and fingerprints requires $n^{O(t)}$ time. The outer code construction, dictionary lookup, ranking, and unranking require polynomial time in $n$ for fixed $r,t$ by Theorem~\ref{thm:pid-outer-partition}. This establishes the uniform complexity claim. Finally, the code redundancy evaluates to $\red\bigl(\Enc^{\mathsf x}_{\boldsymbol m,t}\bigr)=\sum_{j=1}^{r}\log(n+j)=r\log n+O_t(1)$, establishing \eqref{eq:red_mp}.
\end{proof}

\subsection{Regular Profiles and the PID-to-SID Transition}
\label{mp:subsec:regular}

Subsection~\ref{mp:subsec:systematic} appends singleton marker classes and therefore does not preserve regularity when message multiplicities exceed one. We now require both the message and the codeword to be strictly $\lambda$-regular. For integers $k,\lambda\ge1$, we denote
\begin{equation*}
\cM_{k,\lambda} \triangleq \cM(\underbrace{\lambda,\ldots,\lambda}_{k\text{ times}}),
 \qquad n \triangleq k\lambda.
\end{equation*}
A strictly $\lambda$-regular systematic encoder maps $\cM_{k,\lambda}$ to $\cM_{k+s,\lambda}$ with $\Enc(x)|_{[k]}=x$, so the total number of added symbols is $R \triangleq s\lambda$. Hence, $R$ must be divisible by $\lambda$.

We fix $R=s\lambda\ge2t$, partition $[R]$ into consecutive blocks $J_b \triangleq \{(b-1)\lambda+1,\ldots,b\lambda\}$ for $b\in[s]$, and define the domain
\begin{equation*}
\cZ^{\mathrm{reg}}_{n,R,\lambda} \triangleq \bigl\{z\in\{0,\ldots,n\}^{R} : z_{(b-1)\lambda+1}\le\cdots\le z_{b\lambda} \text{ for every } b\in[s]\bigr\}.
\end{equation*}
We restrict the outer syndrome $\Omega_{n,R,t}$ from Theorem~\ref{thm:pid-outer-partition} to this domain, and let $\cA^{\mathrm{reg}}_{n,R,t,\lambda}$ denote its lexicographically first largest syndrome class. For $z\in\cZ^{\mathrm{reg}}_{n,R,\lambda}$ and $x\in\cM_{k,\lambda}$, we define $\Int^{\mathrm{reg}}_z(x)$ from $W_n(z)$ by replacing the $i$-th copy of $D$ with $x_i$, and replacing each formal marker $j\in J_b$ with the physical symbol $k+b$.

\begin{lemma}[Regular Outer Realization]\label{mp:lem:regular-outer-realization}
For fixed parameters $R,t,\lambda$, we have $|\cA^{\mathrm{reg}}_{n,R,t,\lambda}|=\Omega_{R,t,\lambda}(n^{R-t})$, and distinct vectors in this class have disjoint projected output sets under $\cT_{\le t}$. Furthermore,
\begin{equation}
\Int^{\mathrm{reg}}_z(x)\in\cM_{k+s,\lambda},\quad
\Int^{\mathrm{reg}}_z(x)|_{[k]}=x,\quad \Lift(\Int^{\mathrm{reg}}_z(x))=\Int_z(\Lift(x)). \label{mp:eq:regular-lift-interleaving}
\end{equation}
If $y$ is obtained from $\Int^{\mathrm{reg}}_z(x)$ after at most $t$ SIDs or PIDs, then
\begin{equation}
\Theta_n(\Lift(y))\in\cT_{\le t}(z).
 \label{mp:eq:regular-outer-projection}
\end{equation}
The code class, its dictionary, and ranking and unranking within it are computable in polynomial time in $n$ for fixed $R,t,\lambda$.
\end{lemma}

\begin{proof}
Each block of $\lambda$ coordinates is a non-decreasing tuple from $\{0,\ldots,n\}$, which gives $|\cZ^{\mathrm{reg}}_{n,R,\lambda}|=\binom{n+\lambda}{\lambda}^{s}=\Theta_{R,\lambda}(n^R)$. Because Theorem~\ref{thm:pid-outer-partition} employs only $O_{R,t}(n^t)$ outer syndromes, its largest restricted class has size $|\cA^{\mathrm{reg}}_{n,R,t,\lambda}|=\Omega_{R,t,\lambda}(n^{R-t})$. Disjointness and polynomial computability follow from the properties of the outer partition. The block monotonicity and the tie-breaking convention in $W_n(z)$ ensure that the $\lambda$ occurrences of symbol $k+b$ lift precisely to the formal labels in $J_b$, proving \eqref{mp:eq:regular-lift-interleaving}. Finally, the stable lift of a PID output is an ordinary permutation-PID output by Lemma~\ref{mp:lem:stable-lift}. The same holds for an SID output because $\Lift(v)=\Lift(\mstd(v))$. Equation \eqref{mp:eq:regular-outer-projection} then follows from Lemma~\ref{lem:pid-projection}.
\end{proof}

\begin{lemma}[Common Regular PID Output]
\label{mp:lem:regular-common-data-pid}
If two sequences $\Int^{\mathrm{reg}}_z(x)$ and $\Int^{\mathrm{reg}}_{z'}(x')$ produce a common output after $d$ PIDs, then $x$ and $x'$ share a common multipermutation PID descendant after at most $d$ deletions.
\end{lemma}

\begin{proof}
In the two explanations of the common standardized output, let $\delta,\delta'$ denote the numbers of deleted data occurrences, and let $a,a'$ denote the numbers of surviving data classes. Assume $a\le a'$ and retain only the first $a$ symbol classes of the common output. Under the first explanation, this yields the standardized data sequence after $\delta$ deletions. Under the second explanation, retaining the first $a'$ classes and subsequently deleting the entries in classes $a+1,\ldots,a'$ eliminates $(n-\delta')-(n-\delta)=\delta-\delta'$ entries. By transitivity of $\mstd$, the resulting word is identical, representing a $\delta$-PID descendant of both messages. Because $\delta\le d$, the claim follows.
\end{proof}

For a regular composition profile, define $\eta_{\lambda,t} \triangleq \left\lfloor t/\lambda\right\rfloor$. When $\lambda\ge2$, Theorems~\ref{mp:thm:sid-inner} and~\ref{mp:thm:pid-inner-direct} yield $O_t(n^{4t-1})$ and $O_{t,\lambda}(n^{4t-1+2\eta_{\lambda,t}})$ inner syndrome values, respectively. We define the total number of redundancy symbols as
\begin{equation*}
 R_{\SID}^{\mathrm{reg}}(\lambda,t) \triangleq
 \begin{cases}
  4t, & \lambda=1,\\[1mm]
  \lambda\left\lceil\dfrac{5t}{\lambda}\right\rceil, & \lambda\ge2,
 \end{cases}\qquad
 R_{\PID}^{\mathrm{reg}}(\lambda,t) \triangleq
 \begin{cases}
  7t-1, & \lambda=1,\\[1mm]
  \lambda\left\lceil
       \dfrac{5t+2\lfloor t/\lambda\rfloor}{\lambda}
       \right\rceil, & \lambda\ge2,
 \end{cases}
\end{equation*}
and set $s_{\mathsf x}^{\mathrm{reg}} \triangleq R_{\mathsf x}^{\mathrm{reg}}/\lambda$. For $\lambda\ge2$, Lemma~\ref{mp:lem:regular-outer-realization} and the inequalities $R_{\SID}^{\mathrm{reg}}-t>4t-1$ and $R_{\PID}^{\mathrm{reg}}-t>4t-1+2\eta_{\lambda,t}$ demonstrate that the regular outer class contains sufficiently many vectors to index all inner syndromes. Ordering that class lexicographically, we map the inner syndrome index of $x$ to a vector $z$ of the same rank, and define
\begin{equation}
 \Enc^{\mathsf x,\mathrm{reg}}_{k,\lambda,t}(x) \triangleq \Int^{\mathrm{reg}}_z(x)\in\cM_{k+s_{\mathsf x}^{\mathrm{reg}},\lambda}.
 \label{mp:eq:strict-regular-encoder}
\end{equation}
For $\lambda=1$, we employ the permutation encoders developed in Sections~\ref{sec:pid} and~\ref{sec:sid}.

\begin{theorem}[Strictly Regular Systematic Multipermutation Codes]
\label{mp:thm:strict-regular}
Fix integers $t,\lambda\ge1$. For all sufficiently large message alphabet sizes $k$, the encoders defined in \eqref{mp:eq:strict-regular-encoder} are full-systematic and correct up to $t$ SIDs and PIDs, respectively. For fixed $t$ and $\lambda$, they are uniformly constructible, encodable, and decodable in $n^{O_{t,\lambda}(1)}$ time, and their redundancies satisfy
\begin{align*}
\red(\Enc^{\SID,\mathrm{reg}}_{k,\lambda,t}) &= R_{\SID}^{\mathrm{reg}}(\lambda,t)\log n+O_{t,\lambda}(1),\\
\red(\Enc^{\PID,\mathrm{reg}}_{k,\lambda,t}) &= R_{\PID}^{\mathrm{reg}}(\lambda,t)\log n+O_{t,\lambda}(1).
\end{align*}
If $\lambda>t$, the PID and SID error models coincide on every codeword, and the two encoders may be chosen identically with
\begin{equation*}
R_{\PID}^{\mathrm{reg}}(\lambda,t) = R_{\SID}^{\mathrm{reg}}(\lambda,t) = \lambda\left\lceil\frac{5t}{\lambda}\right\rceil.
\end{equation*}
\end{theorem}

\begin{proof}
Lemma~\ref{mp:lem:regular-outer-realization} establishes strict regularity, full systematicity, and unique recovery of the outer vector, and consequently of the inner syndrome. A common SID output restricts to a common data SID output, while Lemma~\ref{mp:lem:regular-common-data-pid} provides the corresponding PID property. The inner syndrome theorems therefore guarantee unique decoding. The decoding algorithm executes the dictionary lookup and subsequently enumerates candidate parents using Lemma~\ref{mp:lem:parent-enumeration}, where all search exponents remain constant for fixed $t$ and $\lambda$.

Next, let $R \triangleq R_{\mathsf x}^{\mathrm{reg}}$ and $s \triangleq R/\lambda$. The code redundancy evaluates to
\begin{equation*}
 \log\frac{|\cM_{k+s,\lambda}|}{|\cM_{k,\lambda}|}
 = \log\frac{(n+R)!}{n!(\lambda!)^s}
 = R\log n+O_{t,\lambda}(1).
\end{equation*}
Finally, if $\lambda>t$, no symbol class can be completely deleted after at most $t$ deletions. Hence, tie-preserving standardization is the identity operation on every surviving sequence, ensuring that $\cD_{\le t}^{\PID}(c)=\cD_{\le t}^{\SID}(c)$ for all codewords $c\in\cM_{k+s,\lambda}$. Because $\eta_{\lambda,t}=0$, the redundancy parameters $R$ coincide, allowing the same inner syndrome and outer class to be utilized.
\end{proof}

As an illustrative example, when $t=2$, the redundancy parameters evaluate to
\begin{equation*}
\begin{array}{c|ccc}
 \lambda & 1 & 2 & \lambda\ge3\\ \hline
 R_{\SID}^{\mathrm{reg}}
   & 8 & 10 & \lambda\lceil10/\lambda\rceil\\
 R_{\PID}^{\mathrm{reg}}
   & 13 & 12 & \lambda\lceil10/\lambda\rceil
\end{array}
\end{equation*}
Thus, the PID redundancy penalty vanishes completely for all $\lambda\ge3$, and the remaining variation arises solely from the rounding required to ensure that the number of added symbols is divisible by $\lambda$.

%% file: main_paper/comparison.tex
\section{Generality of the Outer Framework and Redundancy Comparisons}
\label{sec:comparison}

In this section, we establish the generality of the labeled-marker outer coding framework as a universal systematic compiler for permutation deletion codes. We demonstrate that this framework can transform any inner syndrome partition into a full-systematic code, provided the syndrome classes correct the prescribed deletions. To quantify the reduction in redundancy achieved by the algebraic techniques introduced in Sections~\ref{sec:pid} and~\ref{sec:sid}, we construct baseline systematic codes via direct local algebraic fingerprinting on the full confusability graphs. Comparing these baseline codes with our proposed constructions highlights the essential role played by the cyclic-successor Hamming reduction for SIDs and the single-PID VT decomposition for PIDs. Furthermore, we evaluate the special case of single-deletion correction and provide a comprehensive comparison of redundancy and decoding complexity.

\subsection{The Outer-Code Framework as a Systematic Compiler}
\label{subsec:comparison-outer-framework}

Let $\mathsf x\in\{\SID,\PID\}$. Suppose that
\begin{equation}
 \widetilde{\kappa}_{n,t}^{\mathsf x} \colon \mathcal S_n \longrightarrow \{0,\ldots,\widetilde{M}_{n,t}^{\mathsf x}-1\}
 \label{eq:comparison-generic-inner}
\end{equation}
is an indexed inner syndrome mapping whose syndrome classes each correct up to $t$ $\mathsf x$-deletions in $\mathcal S_n$. The syndrome value identifies the subcode containing a message permutation, and knowledge of this subcode allows the message to be recovered unambiguously from any valid deletion descendant. The outer coding framework developed in Section~\ref{subsec:pid-outer} encodes this syndrome into the relative positions of distinct redundancy markers, yielding a full-systematic codeword in $\mathcal S_{n+r}$.

For an integer $r\geq2t$, Theorem~\ref{thm:pid-outer-partition} guarantees the existence of an outer code class $\mathcal A_{n,r,t}\subseteq\{0,\ldots,n\}^r$ with size
\begin{equation}
 |\mathcal A_{n,r,t}| \geq \frac{(n+1)^{r-t}}{B_{r,t}Q_{r,t}} = \Omega_{r,t}(n^{r-t}).
 \label{eq:comparison-outer-size}
\end{equation}
By Lemma~\ref{lem:sid-to-outer-projection}, this outer class corrects the projected outer channel induced by either deletion model. Assume that the syndrome cardinality satisfies $\widetilde{M}_{n,t}^{\mathsf x}\leq|\mathcal A_{n,r,t}|$. We associate each syndrome index with the vector of identical rank in $\mathcal A_{n,r,t}$ under lexicographic ordering. For a message $\mu\in\mathcal S_n$, the general systematic encoder is defined by
\begin{equation}
 \widetilde{\Enc}_{n,r,t}^{\mathsf x}(\mu) \triangleq \Int_z(\mu),
 \quad \text{where } z \triangleq \operatorname{unrank}_{\mathcal A_{n,r,t}}\bigl(\widetilde{\kappa}_{n,t}^{\mathsf x}(\mu)\bigr).
 \label{eq:comparison-generic-encoder}
\end{equation}
Because the interleaving map satisfies $\Int_z(\mu)|_{[n]}=\mu$, the encoder is full-systematic.

The error-correction capability follows the arguments in Theorems~\ref{thm:main-systematic-pid} and~\ref{thm:systematic-sid}. If two messages $\mu,\nu\in\mathcal S_n$ generate codewords sharing a descendant $y$ after at most $t$ deletions, Lemma~\ref{lem:pid-projection} for PIDs and Lemma~\ref{lem:sid-to-outer-projection} for SIDs imply that the corresponding outer vectors share a projected descendant. By Theorem~\ref{thm:pid-outer-partition}, their outer vectors must be identical, which enforces $\widetilde{\kappa}_{n,t}^{\mathsf x}(\mu)=\widetilde{\kappa}_{n,t}^{\mathsf x}(\nu)$. Restricting $y$ to the message symbols then yields a common descendant of $\mu$ and $\nu$ within the same inner syndrome class, which contradicts the correction capability of that class.

If the inner syndrome size satisfies $\widetilde{M}_{n,t}^{\mathsf x} = O_t(n^{a_{\mathrm{in}}})$ for some constant exponent $a_{\mathrm{in}}$, \eqref{eq:comparison-outer-size} guarantees that $|\mathcal A_{n,r,t}|\geq\widetilde{M}_{n,t}^{\mathsf x}$ for all sufficiently large $n$ whenever
\begin{equation}
 r\geq2t \quad \text{and} \quad r-t > a_{\mathrm{in}}.
 \label{eq:comparison-capacity-condition}
\end{equation}
Thus, choosing $r \triangleq \max\{2t, a_{\mathrm{in}}+t+1\}$ redundancy symbols suffices, yielding an overall redundancy of
\begin{equation}
 \red(\widetilde{\Enc}_{n,r,t}^{\mathsf x}) = \log\frac{(n+r)!}{n!} = r\log n+O_t(1).
 \label{eq:comparison-redundancy}
\end{equation}

When $t=1$, the general inner fingerprinting step can be bypassed entirely by utilizing elementary single-deletion syndromes.

\begin{corollary}[Single-Deletion Systematic Codes]
\label{cor:single-deletion-systematic}
For all sufficiently large $n$, there exist full-systematic permutation codes on $\mathcal S_n$ that correct a single PID using 4 redundancy markers, or a single SID using 3 redundancy markers. Their respective redundancies are $4\log n+O(1)$ and $3\log n+O(1)$ bits.
\end{corollary}

\begin{proof}
For PIDs with $t=1$, the signature syndrome $\nu_n$ in \eqref{eq:single-pid-syndrome} takes $n^2$ values, with $a_{\mathrm{in}}=2$. Setting $r=4$ yields $|\mathcal A_{n,4,1}|=\Omega(n^3)>n^2$ for sufficiently large $n$. For SIDs with $t=1$, the non-binary Varshamov-Tenengolts syndrome restricted to $\mathcal S_n$ takes at most $n$ values, with $a_{\mathrm{in}}=1$. Setting $r=3$ yields $|\mathcal A_{n,3,1}|=\Omega(n^2)>n$. In both cases, composing these syndromes with the outer class via \eqref{eq:comparison-generic-encoder} produces the claimed codes.
\end{proof}

\subsection{Baseline Constructions via Direct Confusability Graph Coloring}
\label{subsec:direct-baselines}

To demonstrate the quantitative gains achieved by our algebraic inner constructions, we consider baseline codes obtained by applying the two-stage local fingerprinting technique of Lemma~\ref{lem:local-fingerprint} directly to the unconstrained confusability graphs $G_{n,t}^{\SID}$ and $G_{n,t}^{\PID}$.

In the SID confusability graph $G_{n,t}^{\SID}$, each permutation has $\binom nt$ descendants, each of which has $t!\binom nt$ parents by Lemma~\ref{lem:parent-enumeration}. Hence, the maximum degree is bounded by
\begin{equation}
 K_{n,t}^{\SID,\mathrm{base}} \triangleq t!\binom nt^{\!2} = O_t(n^{2t}).
 \label{eq:direct-sid-degree}
\end{equation}
Applying the two-stage local fingerprinting of Lemma~\ref{lem:local-fingerprint} directly to $G_{n,t}^{\SID}$ starting from the Lehmer rank produces an inner syndrome $\kappa_{n,t}^{\SID,\mathrm{base}}$ taking at most
\begin{equation*}
 M_{n,t}^{\SID,\mathrm{base}} = O_t\bigl((K_{n,t}^{\SID,\mathrm{base}})^2\bigr) = O_t(n^{4t})
\end{equation*}
distinct values. The exponent is $a_{\mathrm{in}}=4t$. According to the compiler condition in \eqref{eq:comparison-capacity-condition}, this direct approach requires $r = 5t+1$ redundancy markers, achieving a systematic redundancy of
\begin{equation*}
 \red(\Enc_{n,t}^{\SID,\mathrm{base}}) = (5t+1)\log n+O_t(1).
\end{equation*}

In the PID confusability graph $G_{n,t}^{\PID}$, standardization conceals the values of deleted symbols. By Lemma~\ref{lem:parent-enumeration}, each descendant has at most $t!\binom nt^2$ parents, yielding a maximum degree of
\begin{equation}
 K_{n,t}^{\PID,\mathrm{base}} \triangleq t!\binom nt^{\!3} = O_t(n^{3t}).
 \label{eq:direct-pid-degree}
\end{equation}
Applying Lemma~\ref{lem:local-fingerprint} directly to $G_{n,t}^{\PID}$ yields an inner syndrome of cardinality
\begin{equation*}
 M_{n,t}^{\PID,\mathrm{base}} = O_t\bigl((K_{n,t}^{\PID,\mathrm{base}})^2\bigr) = O_t(n^{6t}),
\end{equation*}
corresponding to $a_{\mathrm{in}}=6t$. The outer compiler therefore requires $r = 7t+1$ markers, leading to a redundancy of
\begin{equation*}
 \red(\Enc_{n,t}^{\PID,\mathrm{base}}) = (7t+1)\log n+O_t(1).
\end{equation*}

\subsection{Impact of Inner Structural Constraints}
\label{subsec:inner-impact}

The preceding baseline derivations reveal that the performance of the full-systematic code is governed by the inner syndrome exponent $a_{\mathrm{in}}$. The analytical constructions introduced in Sections~\ref{sec:pid} and~\ref{sec:sid} significantly improve upon these unconstrained graph colorings.

For SIDs, rather than analyzing deletion overlap directly on permutations, the cyclic-successor mapping in \eqref{eq:sid-successor-map} embeds the Ulam metric into the Hamming metric via $d_{\mathrm H}(f(\bar\pi),f(\bar\sigma))\le 3d_{\mathrm U}(\pi,\sigma)$. This reduction allows us to bypass graph coloring entirely and employ the canonical power-sum Hamming syndrome $J_{n+1,3t+1}$. This algebraic syndrome has size $O_t(n^{3t-1})$, reducing the syndrome exponent from $4t$ down to $3t-1$. Consequently, the required number of markers decreases from $5t+1$ to $4t$, saving $(t+1)\log n$ bits of redundancy.

For PIDs, the single-PID syndrome $\nu_n$ partitions $\mathcal S_n$ into subcodes such that each exact-$t$ descendant has at most $L_{n,t} = O_t(n^{2t-2})$ compatible parents in the same subcode, compared to $O_t(n^{2t})$ parents in the unconstrained space. This decomposes the full confusability graph into residual subgraphs whose maximum degree drops from $O_t(n^{3t})$ to $O_t(n^{3t-2})$. Applying local fingerprinting within each residual graph produces a secondary syndrome of size $O_t(n^{6t-4})$. Accounting for the $n^2$ values of the base syndrome $\nu_n$, the total inner syndrome size is $O_t(n^{6t-2})$, reducing the exponent from $6t$ to $6t-2$. This algebraic preprocessing saves two redundancy markers, lowering the redundancy from $(7t+1)\log n$ to $(7t-1)\log n$.

\begin{table*}[!t]
\centering
\caption{Comparison of Inner Syndrome Properties and Resulting Full-Systematic Codes}
\label{tab:comparison-general}
\begin{tabular*}{\textwidth}{@{\extracolsep{\fill}}lllccc@{}}
\toprule
Model & Inner Construction Technique & Inner Syndrome Size & Markers $r$ & Systematic Redundancy & Decoder Complexity \\
\midrule
\multirow{2}{*}{SID} 
 & Direct Graph Fingerprinting (Baseline) & $O_t(n^{4t})$ & $5t+1$ & $(5t+1)\log n + O_t(1)$ & $n^{O(t)}$ \\
 & Cyclic Successor and Hamming Partition & $O_t(n^{3t-1})$ & $4t$ & $4t\log n + O_t(1)$ & $n^{O(t)}$ \\
\midrule
\multirow{2}{*}{PID} 
 & Direct Graph Fingerprinting (Baseline) & $O_t(n^{6t})$ & $7t+1$ & $(7t+1)\log n + O_t(1)$ & $n^{O(t)}$ \\
 & VT Residual Decomposition and Fingerprinting & $O_t(n^{6t-2})$ & $7t-1$ & $(7t-1)\log n + O_t(1)$ & $n^{O(t)}$ \\
\bottomrule
\end{tabular*}
\par\vspace{1ex}
\parbox{\textwidth}{\footnotesize Note: All asymptotic bounds are stated for fixed $t\ge1$ as $n\to\infty$. The number of redundancy markers $r$ is chosen as the minimal integer satisfying $r\ge2t$ and $r-t>a_{\mathrm{in}}$. For the special case $t=1$, the elementary non-binary VT syndrome for SID and the signature syndrome for PID achieve $r=3$ ($3\log n$ redundancy) and $r=4$ ($4\log n$ redundancy), respectively.}
\end{table*}

\subsection{Asymptotic Redundancy and Complexity Trade-Offs}
\label{subsec:comparison-advantages}

Table~\ref{tab:comparison-general} summarizes the inner syndrome parameters and the resulting full-systematic permutation codes. All four code families achieve polynomial encoding and decoding complexity in $n$ for any fixed $t$.

For small deletion values, the quantitative differences are substantial. When $t=1$, Corollary~\ref{cor:single-deletion-systematic} shows that $r=3$ markers suffice for SIDs, giving redundancy $3\log n+O(1)$, while $r=4$ markers suffice for PIDs, giving redundancy $4\log n+O(1)$. When $t=2$, the direct fingerprinting baseline requires $11\log n+O(1)$ redundancy bits for SIDs and $15\log n+O(1)$ bits for PIDs. In contrast, Theorems~\ref{thm:main-systematic-pid} and~\ref{thm:systematic-sid} achieve $8\log n+O(1)$ bits for SIDs and $13\log n+O(1)$ bits for PIDs, eliminating three and two redundancy symbols, respectively.

Regarding computational requirements, evaluating the canonical Hamming syndrome for SIDs requires only $O(nt)$ field operations over $\mathbb F_p$, entirely eliminating graph exploration during encoding. Conversely, both the baseline colorings and the residual PID construction require evaluating local algebraic fingerprints over radius-two graph neighborhoods. Because these neighborhoods are enumerable in $n^{O(t)}$ time, all constructions maintain uniform polynomial-time encodability and decodability.

%% file: main_paper/conclusion.tex
\section{Conclusion}
\label{sec:conclusion}

In this paper, we developed the first full-systematic coding framework for permutations and multipermutations under both SID and PID models. By embedding an inner deletion-correcting syndrome into the relative positions of distinct redundancy markers through an algebraic outer code, our approach circumvents the impossibility of appending conventional checksum suffixes to permutations. For any fixed number of deletions $t$, the proposed constructions achieve redundancies of $(7t-1)\log n+O_t(1)$ bits for PIDs and $4t\log n+O_t(1)$ bits for SIDs, both of which are order-optimal and admit uniform polynomial-time encoding and decoding algorithms. We also extended this two-stage architecture to fixed-composition multipermutations and strictly $\lambda$-regular profiles, characterizing the transition where the PID and SID models coincide when the multiplicity satisfies $\lambda > t$. In addition, our triangle-counting analysis on the confusability graph yields an improved existential lower bound of $\Omega_t(m!\log m/m^{3t})$ for PID permutation codes. Promising directions for future research include closing the gap between existential redundancy bounds and constructive marker counts, as well as designing systematic permutation codes for burst deletions and localized error channels.

%% file: Appendix/lower_bound.tex
\appendices
\section{Proof of Theorem~\ref{thm:prelim-bounds}}
\label{app:proof-thm-bounds}

The SID bounds in Theorem~\ref{thm:prelim-bounds} are already known in \cite{wang2025permutation}. Moreover, every common SID descendant becomes a common PID descendant after standardization, i.e,, $E\bigl(G^{\mathrm{SID}}_{m,t}\bigr)\subseteq E\bigl(G^{\mathrm{PID}}_{m,t}\bigr)$.
Hence, every PID-correcting code is also SID-correcting, and the PID upper bound follows from the SID upper bound as:
\begin{equation*}
 M^{\mathrm{PID}}(m,t)
 \le M^{\mathrm{SID}}(m,t)
 =O_t\left(\frac{m!}{m^t}\right).
\end{equation*}
It remains to prove the following lower bound on the maximum size of PID codes
\begin{equation}
 M^{\mathrm{PID}}(m,t)
 =\Omega_t\!\left(\frac{m!\log m}{m^{3t}}\right).
 \label{eq:pid-lower-target}
\end{equation}

\subsection{Maximum degree and Two-order replacements}

Let $G:=G^{\mathrm{PID}}_{m,t}$, $N:=|V(G)|=m!$, $\Delta:=\Delta(G)$,
and $T$ denote the number of unordered triangles in $G$. By Lemma~\ref{lem:parent-enumeration}, a fixed exact-$t$ PID descendant has at most $t!\binom{m}{t}^{\!2}$ parents in $\mathcal S_m$. Since a permutation has at most $\binom{m}{t}$ exact-$t$ descendants, we have
\begin{equation}
 \Delta\le \binom{m}{t}\left(t!\binom{m}{t}^{\!2}-1\right)=O_t(m^{3t}).\label{eq:pid-degree-app}
\end{equation}
Thus, the only new ingredient needed for the logarithmic improvement is an upper bound on $T$.

A permutation $\omega=(\omega_1,\ldots,\omega_m)\in\mathcal S_m$ may be viewed as $m$ points equipped with two total orders: the \emph{position order}, in which the point at coordinate $i$ has rank $i$, and the \emph{value order}, in which that point has rank $\omega_i$. Equality of standardized subsequences is exactly isomorphism of the induced two-order structures. 

For $u,p,q\in[m]$, define the \emph{one-point replacement} $\Rep_{u,p,q}:\mathcal S_m\longrightarrow\mathcal S_m$
as follows. Starting from $\omega$, delete the point having position rank $u$ and standardize the remaining $m-1$ points. If a surviving point then has position-value rank pair $(a,b)\in[m-1]^2$, replace it by $\bigl(a+\mathbf 1\{a\ge p\},\,b+\mathbf 1\{b\ge q\}\bigr)$, and add a new point with rank pair $(p,q)$. The resulting $m$ rank pairs are the graph of a unique permutation, denoted by $\Rep_{u,p,q}(\omega)$. Thus, $u$ identifies the deleted point, whereas $p$ and $q$ are the position and value ranks of the inserted point.

A replacement sequence of length $h$ is
\begin{equation}
 \mathcal S=((u_1,p_1,q_1),\ldots,(u_h,p_h,q_h))
 \in([m]^3)^h.
 \label{eq:replacement-sequence-app}
\end{equation}
Starting with $\omega_0:=\pi$, define recursively
\begin{equation*}
 \omega_k:=\Rep_{u_k,p_k,q_k}(\omega_{k-1}),
 \qquad k\in[h],
\end{equation*}
and write $\pi\circ\mathcal S:=\omega_h$ for the endpoint of the sequence.

\begin{lemma}[Replacement characterization]
\label{lem:replacement-characterization-app}
For all $\pi,\sigma\in\mathcal S_m$, the minimum number of one-point replacements transforming $\pi$ into $\sigma$ equals $d_{\mathrm P}(\pi,\sigma)$.
\end{lemma}

\begin{proof}
Suppose that $h$ replacements transform $\pi$ into $\sigma$. At most $h$ original points are deleted, and the relative position and value orders among all surviving original points never change. Hence $\pi$ and $\sigma$ contain an order-isomorphic subsequence of length at least $m-h$, so $d_{\mathrm P}(\pi,\sigma)\le h$.

Conversely, let $r:=d_{\mathrm P}(\pi,\sigma)$. By the definition of $d_{\mathrm P}$, there are subsets $C_\pi$ and $C_\sigma$, each of size $m-r$, whose induced two-order structures are isomorphic. Identify these two cores through their unique isomorphism. Keep the points of $C_\pi$ throughout the procedure, and list the target points in $\sigma\setminus C_\sigma$ in any fixed order. At step $j$, delete one point of $\pi\setminus C_\pi$ that is still present and insert the $j$th target point. In each of the two total orders, place this target point in the gap determined by its order relations to the core and to the target points inserted in earlier steps; its position relative to the original points that will be deleted later is arbitrary. Such a gap always exists in each order. Subsequent deletions preserve all prescribed order relations. After $r$ steps, precisely the core and all $r$ target points remain, with the two total orders induced by $\sigma$. Hence $r$ replacements suffice.
\end{proof}

We will track points through a replacement sequence. Label the $m$ initial points by $O:=\{o_1,\ldots,o_m\}$, where $o_i$ is the point initially at position $i$. Label the point inserted at step $k$ by $f_k$, and put $F_h:=\{f_1,\ldots,f_h\}$. These labels are bookkeeping devices and do not alter the underlying permutation. If a labeled point $x$ is present after step $k$, let $\rk_k^{\mathrm{pos}}(x)$ and $\rk_k^{\mathrm{val}}(x)$
denote its position and value ranks at that time. In particular,
\begin{equation*}
 \rk_0^{\mathrm{pos}}(o_i)=i,
 \qquad
 \rk_0^{\mathrm{val}}(o_i)=\pi_i,
\end{equation*}
and the birth ranks of $f_k$ are
\begin{equation*}
 \rk_k^{\mathrm{pos}}(f_k)=p_k,
 \qquad
 \rk_k^{\mathrm{val}}(f_k)=q_k.
\end{equation*}
A replacement not deleting $x$ changes each of its two ranks by at most one: deletion changes a rank by zero or $-1$, and insertion then changes it by zero or $+1$.

For fixed $\pi\in\mathcal S_m$ and $0\le h\le2t$, define
\begin{equation}
 \mathcal R_h(\pi,t)
 :=\{\mathcal S\in([m]^3)^h:
 d_{\mathrm P}(\pi,\pi\circ\mathcal S)\le t\}.
 \label{eq:Rh-app}
\end{equation}
Thus, $\mathcal R_h(\pi,t)$ consists of the length-$h$ replacement sequences whose endpoints remain adjacent to, or equal to, $\pi$ in the exact-$t$ PID graph.

\begin{lemma}[Close-endpoint replacement sequences]
\label{lem:close-endpoint-app}
For fixed $t$, uniformly over $\pi\in\mathcal S_m$ and $0\le h\le2t$,
\begin{equation}
 |\mathcal R_h(\pi,t)|=O_t(m^{t+2h}).
 \label{eq:close-endpoint-app}
\end{equation}
In particular, $|\mathcal R_{2t}(\pi,t)|=O_t(m^{5t})$.
\end{lemma}

\begin{proof}
The case $h=0$ is immediate. Fix $1\le h\le2t$ and $\mathcal S\in\mathcal R_h(\pi,t)$. Let $R(\mathcal S)\subseteq O$ be the set of initial points deleted during $\mathcal S$, and put $s:=|R(\mathcal S)|$. Thus, $s$ is the number of distinct original points deleted along the sequence. The set $R(\mathcal S)$ has at most $m^s$ choices. We count separately for each $s\in\{0,\ldots,h\}$; summing over these $h+1=O_t(1)$ possibilities does not change the final order.

For each step $k$, let $\chi_k$ be the label of the point deleted at that step, and put
\begin{equation*}
 \boldsymbol\chi(\mathcal S):=(\chi_1,\ldots,\chi_h).
\end{equation*}
Once $R(\mathcal S)$ is fixed, every $\chi_k$ is either one of those $s$ original labels or one of the earlier fresh labels $f_1,\ldots,f_{k-1}$. Hence the number of possible deleted-label histories is at most $(s+h)^h=O_t(1)$. Given $R(\mathcal S)$, $\boldsymbol\chi(\mathcal S)$, and all insertion-rank pairs $(p_k,q_k)$, the sequence is reconstructed uniquely by simulation: $u_k$ is the current position rank of the point labeled $\chi_k$. It therefore remains to count the insertion-rank pairs.

If $s\le t$, there are at most $m^{2h}$ choices for the $h$ pairs $(p_k,q_k)$, and therefore at most $O_t(m^{s+2h})=O_t(m^{t+2h})$ sequences.

Assume now that $s>t$, and set $\ell=s-t$. The endpoint $\tau:=\pi\circ\mathcal S$ contains exactly $s$ surviving fresh points. Indeed, $h$ fresh points are inserted, while the $h-s$ deletion steps not deleting original points delete fresh points.

Since $d_{\mathrm P}(\pi,\tau)\le t$, the initial and endpoint two-order structures contain isomorphic subsets of size $m-t$. Order the labels by
\begin{equation*}
 o_1<\cdots<o_m<f_1<\cdots<f_h.
\end{equation*}
Among all pairs of isomorphic $(m-t)$-subsets, choose the pair whose two increasing label lists are lexicographically smallest, first comparing the initial list and then the endpoint list. Denote this canonical pair by $C_0\subseteq O$ and $C_h$, and let $\varphi:C_0\longrightarrow C_h$
be its unique isomorphism preserving both total orders.
The endpoint core $C_h$ omits only $t$ endpoint points. Since $s$ fresh points survive, at least $s-t=\ell$ of them lie in $C_h$. Let $B$ be the $\ell$ such fresh points with smallest birth indices. Write
\begin{equation*}
 B=\{f_k:k\in K\},\qquad |K|=\ell,
\end{equation*}
and, for every $k\in K$, define
\begin{equation*}
 a_k:=\varphi^{-1}(f_k)\in C_0.
\end{equation*}
The set $K$ has only $O_t(1)$ possibilities, while the ordered choices of the distinct points $(a_k)_{k\in K}$ contribute at most $m^\ell$.

Fix $k\in K$. Because $\varphi$ preserves both total orders, $a_k$ and $f_k$ have the same rank inside their respective cores. Each full structure is obtained from its core by inserting exactly $t$ omitted points. Therefore,
\begin{equation}
 \bigl|\rk_0^{\mathrm{pos}}(a_k)-
 \rk_h^{\mathrm{pos}}(f_k)\bigr|\le t,
 \qquad
 \bigl|\rk_0^{\mathrm{val}}(a_k)-
 \rk_h^{\mathrm{val}}(f_k)\bigr|\le t.
 \label{eq:core-rank-app}
\end{equation}
Since $f_k$ survives, each of its two ranks changes by at most one in each of the at most $h-k$ later replacements. Consequently,
\begin{equation*}
 |p_k-\rk_h^{\mathrm{pos}}(f_k)|\le h,
 \qquad
 |q_k-\rk_h^{\mathrm{val}}(f_k)|\le h.
\end{equation*}
Combining these inequalities with \eqref{eq:core-rank-app},
\begin{equation}
 |p_k-\rk_0^{\mathrm{pos}}(a_k)|\le t+h,
 \qquad
 |q_k-\rk_0^{\mathrm{val}}(a_k)|\le t+h.
 \label{eq:birth-rank-app}
\end{equation}
Because $h\le2t$, once $a_k$ is fixed, the pair $(p_k,q_k)$ has at most $(6t+1)^2=O_t(1)$ possibilities. The remaining $h-\ell$ insertion-rank pairs have at most $m^{2(h-\ell)}$ possibilities. The total number of sequences in this case is therefore at most
\begin{equation*}
 O_t(1)m^s m^\ell m^{2(h-\ell)}
 =O_t(m^{s-\ell+2h})
 =O_t(m^{t+2h}),
\end{equation*}
because $\ell=s-t$. This proves \eqref{eq:close-endpoint-app}.
\end{proof}

\subsection{Triangle count}

For every ordered pair $(\pi,\sigma)$ with $d_{\mathrm P}(\pi,\sigma)\le t$, choose the lexicographically first shortest replacement sequence from $\pi$ to $\sigma$, and denote it by $\mathsf P(\pi,\sigma)$. By Lemma~\ref{lem:replacement-characterization-app}, $|\mathsf P(\pi,\sigma)|=d_{\mathrm P}(\pi,\sigma)\le t$. The following canonical-path encoding is the PID analogue of the translocation-path encoding in \cite[Proof of Lemma.~4]{wang2025permutation}.

Fix a root $\pi$, and let $(\pi,\sigma,\tau)$ be an ordered triangle of $G$. Concatenate $\mathsf P(\pi,\sigma)$ and $\mathsf P(\sigma,\tau)$. Let
\[
 h_1:=|\mathsf P(\pi,\sigma)|,
 \qquad
 h_2:=|\mathsf P(\sigma,\tau)|,
 \qquad
 h:=h_1+h_2.
\]
Then $h\le2t$, and the third edge of the triangle gives $d_{\mathrm P}(\pi,\tau)\le t$. Hence the concatenated sequence lies in $\mathcal R_h(\pi,t)$. The map
\[
 (\sigma,\tau)
 \longmapsto
 \bigl(\mathsf P(\pi,\sigma)\Vert
 \mathsf P(\sigma,\tau),h_1\bigr)
\]
is injective: the first $h_1$ operations recover $\sigma$, while all $h$ operations recover $\tau$. For a fixed $h$, the split point $h_1$ has at most $h+1\le2t+1$ possibilities. Therefore, Lemma~\ref{lem:close-endpoint-app} implies that the number of ordered triangles rooted at $\pi$ is
\[
 O_t\!\left(\sum_{h=0}^{2t}m^{t+2h}\right)
 =O_t(m^{5t}).
\]
Summing over $m!$ roots and dividing by six yields
\begin{equation}
 T=O_t(m!m^{5t})=O_t(Nm^{5t}).
 \label{eq:pid-triangle-app}
\end{equation}

We now apply the same few-triangles independence bound used in \cite[Lemma.~3]{wang2025permutation}. If $T=0$, its triangle-free specialization already gives the claimed order, so assume $T>0$. The cited lemma states that
\begin{equation}
 \alpha(G)\ge
 \frac{N}{10\Delta}
 \left(\log\Delta-\frac12\log\frac{T}{N}\right)
 \label{eq:bollobas-app}
\end{equation}
whenever the expression in parentheses is positive.

Choose constants $A_t,B_t>0$ such that, by \eqref{eq:pid-degree-app} and \eqref{eq:pid-triangle-app}, we have $\Delta\le A_tm^{3t}$ and $T\le B_tNm^{5t}$. If $\Delta\le m^{3t}/\log m$, the greedy bound gives
\begin{equation*}
 \alpha(G)\ge\frac{N}{\Delta+1}
 =\Omega_t\!\left(\frac{N\log m}{m^{3t}}\right).
\end{equation*}
Otherwise, $\Delta>m^{3t}/\log m$, and for all sufficiently large $m$,
\begin{align*}
 \log\Delta-\frac12\log\frac{T}{N}
 &\ge 3t\log m-\log\log m
     -\frac12\log(B_tm^{5t})\\
 &=\frac{t}{2}\log m-\log\log m-O_t(1)\\
 &=\Omega_t(\log m).
\end{align*}
Applying \eqref{eq:bollobas-app} and $\Delta\le A_tm^{3t}$ again gives
\begin{equation*}
 \alpha(G)
 =\Omega_t\!\left(\frac{N\log m}{m^{3t}}\right).
\end{equation*}
Since $M^{\mathrm{PID}}(m,t)=\alpha(G)$, this proves \eqref{eq:pid-lower-target} and completes the proof of Theorem~\ref{thm:prelim-bounds}.

%% file: main.bbl
\begin{thebibliography}{10}
\providecommand{\url}[1]{#1}
\csname url@samestyle\endcsname
\providecommand{\newblock}{\relax}
\providecommand{\bibinfo}[2]{#2}
\providecommand{\BIBentrySTDinterwordspacing}{\spaceskip=0pt\relax}
\providecommand{\BIBentryALTinterwordstretchfactor}{4}
\providecommand{\BIBentryALTinterwordspacing}{\spaceskip=\fontdimen2\font plus
\BIBentryALTinterwordstretchfactor\fontdimen3\font minus \fontdimen4\font\relax}
\providecommand{\BIBforeignlanguage}[2]{{%
\expandafter\ifx\csname l@#1\endcsname\relax
\typeout{** WARNING: IEEEtran.bst: No hyphenation pattern has been}%
\typeout{** loaded for the language `#1'. Using the pattern for}%
\typeout{** the default language instead.}%
\else
\language=\csname l@#1\endcsname
\fi
#2}}
\providecommand{\BIBdecl}{\relax}
\BIBdecl

\bibitem{jiang2009rank}
A.~Jiang, R.~Mateescu, M.~Schwartz, and J.~Bruck, ``Rank modulation for flash memories,'' \emph{IEEE Transactions on Information Theory}, vol.~55, no.~6, pp. 2659--2673, 2009.

\bibitem{jiang2010correcting}
A.~Jiang, M.~Schwartz, and J.~Bruck, ``Correcting charge-constrained errors in the rank-modulation scheme,'' \emph{IEEE Transactions on Information Theory}, vol.~56, no.~5, pp. 2112--2120, 2010.

\bibitem{barg2010codes}
A.~Barg and A.~Mazumdar, ``Codes in permutations and error correction for rank modulation,'' \emph{IEEE Transactions on Information Theory}, vol.~56, no.~7, pp. 3158--3165, 2010.

\bibitem{gabrys2015codes}
R.~Gabrys, E.~Yaakobi, F.~Farnoud, F.~Sala, J.~Bruck, and L.~Dolecek, ``Codes correcting erasures and deletions for rank modulation,'' \emph{IEEE Transactions on Information Theory}, vol.~62, no.~1, pp. 136--150, 2016.

\bibitem{farnoud2013error}
F.~Farnoud, V.~Skachek, and O.~Milenkovic, ``Error-correction in flash memories via codes in the {Ulam} metric,'' \emph{IEEE Transactions on Information Theory}, vol.~59, no.~5, pp. 3003--3020, 2013.

\bibitem{levenshtein1966binary}
V.~I. Levenshtein, ``Binary codes capable of correcting deletions, insertions, and reversals,'' \emph{Soviet Physics Doklady}, vol.~10, no.~8, pp. 707--710, 1966.

\bibitem{tenengolts1984nonbinary}
G.~Tenengolts, ``Nonbinary codes, correcting single deletion or insertion,'' \emph{IEEE Transactions on Information Theory}, vol.~30, no.~5, pp. 766--769, 1984.

\bibitem{sima2020optimalbinary}
J.~Sima, R.~Gabrys, and J.~Bruck, ``Optimal systematic $t$-deletion correcting codes,'' in \emph{2020 IEEE International Symposium on Information Theory (ISIT)}.\hskip 1em plus 0.5em minus 0.4em\relax IEEE, 2020, pp. 769--774.

\bibitem{guruswami2021two}
V.~Guruswami and J.~H{\aa}stad, ``Explicit two-deletion codes with redundancy matching the existential bound,'' \emph{IEEE Transactions on Information Theory}, vol.~67, no.~10, pp. 6384--6394, 2021.

\bibitem{levenshtein1992perfect}
V.~I. Levenshtein, ``On perfect codes in deletion and insertion metric,'' \emph{Discrete Mathematics and Applications}, vol.~2, no.~3, pp. 241--258, 1992.

\bibitem{chee2019burst}
Y.~M. Chee, S.~Ling, T.~T. Nguyen, V.~K. Vu, H.~Wei, and X.~Zhang, ``Burst-deletion-correcting codes for permutations and multipermutations,'' \emph{IEEE Transactions on Information Theory}, vol.~66, no.~2, pp. 957--969, 2020.

\bibitem{sun2022improved}
Y.~Sun, Y.~Zhang, and G.~Ge, ``Improved constructions of permutation and multi-permutation codes correcting a burst of stable deletions,'' \emph{IEEE Transactions on Information Theory}, vol.~69, no.~7, pp. 4429--4441, 2023.

\bibitem{wang2025permutation}
S.~Wang, T.~Nguyen, Y.~M. Chee, and V.~K. Vu, ``Permutation and multi-permutation codes correcting multiple deletions,'' \emph{IEEE Transactions on Information Theory}, vol.~71, no.~9, pp. 6759--6770, 2025.

\bibitem{zhou2014systematic}
H.~Zhou, M.~Schwartz, A.~A. Jiang, and J.~Bruck, ``Systematic error-correcting codes for rank modulation,'' \emph{IEEE Transactions on Information Theory}, vol.~61, no.~1, pp. 17--32, 2015.

\bibitem{buzaglo2014perfect}
S.~Buzaglo and T.~Etzion, ``Perfect permutation codes with the {Kendall}'s $\tau$-metric,'' in \emph{2014 IEEE International Symposium on Information Theory (ISIT)}.\hskip 1em plus 0.5em minus 0.4em\relax IEEE, 2014, pp. 2391--2395.

\bibitem{klove2010permutation}
T.~Kl{\o}ve, T.-T. Lin, S.-C. Tsai, and W.-G. Tzeng, ``Permutation arrays under the {Chebyshev} distance,'' \emph{IEEE Transactions on Information Theory}, vol.~56, no.~6, pp. 2611--2617, 2010.

\bibitem{li2026distributed}
Y.~Li, R.~Gabrys, and F.~Farnoud, ``Constructing low-redundancy codes via distributed graph coloring,'' \emph{IEEE Transactions on Information Theory}, vol.~72, no.~10, pp. 7344--7358, Oct. 2026.

\bibitem{gad2012trade}
E.~E. Gad, A.~Jiang, and J.~Bruck, ``Trade-offs between instantaneous and total capacity in multi-cell flash memories,'' in \emph{2012 IEEE International Symposium on Information Theory Proceedings}.\hskip 1em plus 0.5em minus 0.4em\relax IEEE, 2012, pp. 990--994.

\bibitem{huczynska2006frequency}
S.~Huczynska and G.~L. Mullen, ``Frequency permutation arrays,'' \emph{Journal of Combinatorial Designs}, vol.~14, no.~6, pp. 463--478, 2006.

\bibitem{hassanzadeh2014multipermutation}
F.~Farnoud~(Hassanzadeh) and O.~Milenkovic, ``Multipermutation codes in the {Ulam} metric for nonvolatile memories,'' \emph{IEEE Journal on Selected Areas in Communications}, vol.~32, no.~5, pp. 919--932, 2014.

\bibitem{sala2014deletions}
F.~Sala, R.~Gabrys, and L.~Dolecek, ``Deletions in multipermutations,'' in \emph{2014 IEEE International Symposium on Information Theory}.\hskip 1em plus 0.5em minus 0.4em\relax IEEE, 2014, pp. 2769--2773.

\end{thebibliography}
